\documentclass[lettersize,journal]{IEEEtran}
\usepackage{amsmath,amsfonts}
\usepackage{algorithmic}
\usepackage{array}
\usepackage[caption=false,font=normalsize,labelfont=sf,textfont=sf]{subfig}
\usepackage{textcomp}
\usepackage{stfloats}
\usepackage{url}
\usepackage{verbatim}
\usepackage{graphicx}
\usepackage{color}
\usepackage{amssymb,mathrsfs,amsmath}
\usepackage{booktabs}
\usepackage{enumitem}
\usepackage{amsthm}
\usepackage{multicol}
\usepackage{multirow}
\usepackage{bigstrut}
\newtheorem{theorem}{Theorem}

\newtheorem{remark}{Remark}
\newtheorem{definition}{Definition}
\newtheorem{proposition}{Proposition}
\usepackage{balance}

\def\BibTeX{{\rm B\kern-.05em{\sc i\kern-.025em b}\kern-.08em
    T\kern-.1667em\lower.7ex\hbox{E}\kern-.125emX}}

\begin{document}
\title{From Target Tracking to Targeting Track — Part I: A Metric for Spatio-Temporal Trajectory Evaluation}
\author{Tiancheng Li, \textit{IEEE Senior Member}, Yan Song, 
Hongqi Fan, Jingdong Chen, \textit{IEEE Fellow}
\thanks{Manuscript created August 2024; \\
This work was supported in part by the National Natural Science Foundation of China under Grants 62422117 and 62201316 %and 62071389, in part by the Natural Science Basic Research Program of Shaanxi Province under Grant 2023JC-XJ-22, 
and in part by the Fundamental Research Funds for the Central Universities. 
\\
Tiancheng Li and Yan Song 
are with the Key Laboratory of Information Fusion Technology (Ministry of Education), School of Automation, Northwestern Polytechnical University, Xi’an 710129, China, E-mail: t.c.li@nwpu.edu.cn, syzx@mail.nwpu.edu.cn. 
Hongqi Fan is with the College of Electronic Science and Technology, National University of Defense Technology, Changsha 410073, China, E-mail: fanhongqi@nudt.edu.cn. 
Jingdong Chen is with the School of Marine Science and Technology,  Northwestern Polytechnical University, Xi’an 710129, China, E-mail: jingdongchen@ieee.org.
}}

\markboth{Journal of \LaTeX\ Class Files,~Vol.~, No.~9, Aug~2024}%
{How to Use the IEEEtran \LaTeX \ Templates}

\maketitle

\begin{abstract}
In the realm of target tracking, performance evaluation %, independent of the choice of tracking algorithms, 
plays a pivotal role in the design, comparison, and analytics of trackers. 
Compared with the traditional trajectory composed of a set of point-estimates obtained by a tracker in the measurement time-series, the trajectory that our series of studies including this paper pursued is given by a curve function of time (FoT). %Fundamentally, any 
%Any particular realization of the trajectory is a sample of the trajectory stochastic process. % defined in the continuous-time domain. %(FoT) 
The trajectory FoT provides complete information of the movement of the target over time  
% such as the length and shape of the trajectory. % 
and can be used to infer the state corresponding to arbitrary time, not only at the measurement time. 
However, there are no metrics available for comparing and evaluating the trajectory FoT. 
%In fact, the trajectory given by a function of time also offers convenience for the long-time state prediction.
%To fill this gap, 
To address this lacuna, we propose a metric denominated as the 
%length integrated 
%smoothness integrated multiple trajectory assignment distance. %(Star-ID), which is suitable
%for comparing continuous-time curvilinear trajectories. 
spatio-temporal-aligned trajectory integral distance (Star-ID). 
The Star-ID associates and aligns the estimated and actual trajectories in the spatio-temporal domain and %differentiates 
distinguishes between the time-aligned and unaligned segments in calculating the spatial divergence including false alarm, miss-detection and localization errors. 
%For the aligned segments, the arc length of the trajectory, an important attribute of the curve, is taken into account in the integral distance while for the unaligned segments, both FA and MD are identified and penalized with regard to their duration. 
The effectiveness of the proposed distance metric and the time-averaged version is validated through theoretical analysis and numerical examples of a single target or multiple targets.
\end{abstract}


\section{Introduction} \label{sec：intro}
%\IEEEPARstart{I}{n} the domain of target tracking, performance evaluation holds paramount significance, independent of the choice of tracking algorithm. It plays a crucial role in the design, comparison, and evaluation of any estimator. At the core of performance evaluation lies the identification of a suitable metric that can effectively quantify the extent of resemblance between the ground truth and the estimated set \cite{ristic2011metric,vu2020complete,rahmathullah2017generalized,milan2016mot16}. The key technical issue in this process is measuring the distance between two sets  \cite{garcia2021time}. Intuitively, the smaller the distance between the ground truth and the estimated set, the more similar they are considered to be. Therefore, designing a proper measure of the distance is crucial for the accurate performance evaluation.

\IEEEPARstart{M}{ulti}-target tracking (MTT)  is an intricate process that entails the sequential estimation of both the cardinality (number of targets) and the kinematic states of multiple targets, where both parameters are potentially time-variant \cite{Bar-Shalom01,Sarkka13book,Vo15mtt}. It has been a key technology in the applications of autonomous driving, guidance and defense systems, traffic control, and robotics. % such as the position and velocity of each target. 
Typically, each target is assigned a unique identification, and the output of the tracking algorithm is manifested as distinguishable temporal sequences of state estimates, referred to as tracks \cite{Zhang2015Trajectory,Vo19msGLMB,garcia2021time,DiLi-24label}. 
Quantifying the similarity between the ground truth and estimated tracks  \cite{ristic2011metric,Barrios17metrics,vu2020complete}, %milan2016mot16
the core for performance evaluation, which needs to grapple with significant ambiguities \cite{Drummond92} holds paramount significance in the design, comparison, and analytics of any tracker \cite{Branlco1999tool,1999Design,Rothrock2000Performance}. Irrespective of the choice of tracking algorithms, the metric plays an essential role in the performance evaluation. Arguably, different metrics or even different parameter choices 
will all lead to diverse results \cite{Nguyen23Trustworthy}.  
%which needs to be objective and easily measurable. %measures that are essential for algorithm comparison, evaluation and design. 
%The choice of evaluation metric heavily influences these processes. As %a result, the evaluation of target tracking performance has garnered significant attention and research efforts.

%\subsection{Literature Review}
\subsection{Relevant Work}
The most widely used metric for multi-target estimate evaluation is the so-called optimal sub-pattern assignment (OSPA) distance \cite{schuhmacher2008consistent}. This metric extends the Wasserstein
metric to the general case where two point-sets may have different numbers of points. It comprises two components, referred to as the localization error and cardinality/quantity error; the latter accounts for the false alarm (FA) and missed detection (MD), which are two notable challenging issues for MTT. A meaningful extension of this approach is to incorporate the uncertainty of each point-estimate \cite{Nagappa2011ospaUncertainty}. %, where the Hellinger distance between two points, each augmented with corresponding covariance matrices, is used as the base distance for the OSPA. 
Relevantly, the target existence probability taken as the track quality information is accounted for within the OSPA in \cite{he2013track}. Furthermore, 
the negative log-likelihood of the posterior that incorporates all uncertainties is defined as a measure \cite{Pinto2021NLL}. 
The OSPA metric has inspired many other metrics such as the %unnormalized OSAP (UOSPA) \cite{Williams15}, 
generalized OSPA (GOSPA) \cite{rahmathullah2017generalized, Angel2020GOSPA}, and complete OSPA (COSPA) \cite{vu2020complete}, %and $\text{OSPA}^{(2)}$ \cite{beard2017ospa,Beard20LMB}, 
among others \cite{Branko2010ospaLabeld,vu2014new}. %, which overcome one or another limitation of the OSPA. 
%In particular, the $\text{OSPA}^{(2)}$ metric measures the distance between two sets of tracks (of time-series point-estimates) which does not only capture target state errors, cardinality errors, but also phenomena such as track switching and track fragmentation. This representative metric on the track will be detailed next.

%There are a plenty of metrics for comparing and evaluating finite point sets where the set consists of a number of points \cite{ristic2011metric,beard2017ospa}; see section \ref{sec:pre} for the state-of-the-art approaches. 
However, calculating the estimation error at each sampling step in isolation is inadequate for tracker evaluation. Instead, the metric should account for the dissimilarity between the overall tracks consisting of time-series of point-estimates \cite{Bento16arXiv,Beard20LMB,Nguyen23Trustworthy} or of sub-densities \cite{DiLi-24label}; see a substantial body of work reviewed in \cite{hu2023spatio,su2020survey,Nguyen23Trustworthy}. %, which are generally of different lengths 
This line of thought leads to the foundation of the metric of $\text{OSPA}^{(2)}$ \cite{beard2017ospa,Beard20LMB}, which is a spatio-temporal metric that accounts for both the spatial and temporal information of the tracks. Notably, the metric proposed in  \cite{Bento16arXiv} allows trajectory association to be changed over time and incorporates the confusion of trajectories' identity in an optimal way. %has received increasing attention in the MTT community. 
%; see section \ref{sec:pre}.B. 
Nonetheless, these OSPA-like metrics are predominantly designed for use with traditional recursive filters/trackers where the estimates are obtained in discrete measurement time series and so the evaluation deals with point-state estimates obtained in time series as illustrated in Fig. \ref{fig:diff_track} (a) and (b). They are not applicable to the continuous-time trajectory case--with the exception that the metric proposed in \cite{Bento16arXiv} is claimed, yet not validated, to have the potential for extension to continuous time -- where each trajectory is given by a continuous-time curve as shown in Fig. \ref{fig:diff_track} (c), not a set of points obtained in discrete time series. Although there are a large number of studies on continuous-time state estimation \cite{Talbot2024continuoustimestate}, there is to date no metric available for general continuous-time trajectories except for our preliminary work \cite{xin2022metric} as to be explained in this paper. %These methods will be studied later on in this paper.

%In multi-target tracking (MTT), the aim is to estimate the set of target tracks over a period of time, rather than the set of target states at each time step. In this paper, we demonstrate that MTT error can be captured using the OSPA metric to define a distance between two sets of tracks.

\subsection{Introduction to Companion Papers}
Our series of studies are founded on modeling the target trajectory using a curve function of time (FoT) as shown in Fig. \ref{fig:diff_track} (c), which was termed the trajectory FoT (T-FoT) \cite{li2018joint,li2023target}. %, rather than by any discrete-time Markov model. 
That is, the evolution of the target state over time is modeled by T-FoT $f:\mathbb{R}^+ \rightarrow \mathcal{X}$ in the spatio-temporal space and %$\mathbf{x}_t\in \mathcal{X}$ represents 
the target state at time $t$ is given by  
\begin{equation} \label{eq:T-FoT}
{\mathbf{x}_t} = f(t), 
\end{equation}
where $t\in \mathbb{R}^+$ denotes the time and $\mathcal{X}$ denotes the state space.

Specifically, the T-FoT can be given in the forms of a polynomial \cite{li2018joint,li2023target} (see also the companion paper \cite{Li25TFoT-part2}), B-spline basis functions \cite{Furgale12} and so on \cite{Pacholska20}; see the review given in \cite{Talbot2024continuoustimestate}. That being said, the metric we seek in this paper is unlimited to any specific form but applicable for the general continuous-time case. %i.e., 
% \begin{equation}\label{eq:polynomial}
% F\left ( t;\mathbf{C} \right )=\mathbf{c}_{0}+\mathbf{c}_{1}t+\mathbf{c}_{2}t^2+\cdots +\mathbf{c}_{\gamma }t^{\gamma },
% \end{equation}
% where $\gamma $ refers to  the order of the fitting function which may be given exactly in advance or specified with a higher bound, $\mathbf{c}_i= \big\{c_i^{(1)}, c_i^{(2)}, \cdots, c_i^{(r)}  \big\}$, $r$ indicates the dimension of the state space $\mathcal{X}$, and $\mathbf{c}_0,\mathbf{c}_1,\mathbf{c}_2$ corresponds to the initial position, velocity and acceleration of the target \cite{Li25TFoT-part2}. 
%It has been addressed in detail in the companion paper.
% The T-FoT may be modeled in some other forms, e.g., \cite{Pacholska20}. 
% That being said, the metric we seek in this paper is unlimited to any specific form but applicable for the general case. %r, we design the metric for the general case.
% The T-FoT may be modeled in some other forms, e.g., . That being said, the metric we seek in this paper is unlimited to any specific form but applicable for the general case. 
Most existing studies, however, do not provide uncertainty about their T-FoT estimate. Moreover, they fail to utilize any information regarding the state correlation over time. To put it simply, the states corresponding to close time instants are more likely to be in closer spatial proximity to each other than those far away. To take advantage of this latent state temporal correlation and to provide an assessment of the uncertainty associated with the T-FoT estimate, within our series of companion papers \cite{Li25TFoT-part2,Li25TFoT-part3} %including this one, 
we further model %the T-FoT  $f$ of 
the collection of the target state as a stochastic process (SP) $ \mathcal{F} \triangleq \{\mathbf{x}_t: t\in \mathbb{R}^+ \} $ defined on the continuous time. % and specified by the parameters $\varTheta $. 
That is, any particular T-FoT is a sample path of this SP, i.e., % as described as follows %the continuous time-space series Trajectory Function of Time (T-FoT) satisfies:
\begin{equation}
	f(t)\sim \mathcal{F}. \label{eq:TSP} %\left( m(\cdot); \varTheta \right) 
\end{equation}
%where $m(\cdot)$ is the mean function of the TSP which is an FoT. 

The contributions of our series of companion papers including three parts are structured as follows. 
\begin{itemize}
	\item This paper, which serves as Part I of this series of papers, proposes a metric for evaluating the quality of any T-FoT estimate $\hat{f}$. This provides a distance measure between any two trajectories represented in the form of an FoT. To avoid distracting the reader's attention from our key contribution, this paper will not delve further into the SP model but focus only on the T-FoT metric. %This sets the foundation for the following works. 
	\item Part II \cite{Li25TFoT-part2} and Part III \cite{Li25TFoT-part3} offer solutions to online learning of the SP via decomposing it into a deterministic FoT and the residual SP according to some decomposition approaches \cite{Cramer1961,Urbin2012time}. The former fits the trend of the trajectory by the best polynomial FoT based on regularized optimization %that employs two distinct strategies of regularization.  %, seeking trade-off between the accuracy and simplicity. 
	%\item offers solutions for 
	while the latter approximates the residual SP by %the residual zero-mean SP %hyperparameters $\varTheta $ 
	%for 
	two representative SPs: the Gaussian process (GP) and Student's $t$ process (StP).
\end{itemize}

% However, we consider the case where each element of the set to be compared/evaluated is a trajectory described by a function of continuous time; . 

%The problem we are concerned with can be more formally modeled as follows

% \begin{figure}[htbp]
% \centerline{\includegraphics[width=0.55\linewidth]{SS1.png}}
% \caption{A single-target scenario, where \textcolor{red}{\( \bullet\)} represents the real target and \textcolor{blue}{\( \bullet\)} represents the estimates. } \label{fig:diff_track}
% \end{figure}
% \begin{figure}[htbp]
% \centerline{\includegraphics[width=0.55\linewidth]{ss2.png}}
% \caption{A multi-target scenario with two real targets and their corresponding two sets of estimates. } \label{fig11}
% \end{figure}
% \begin{figure}[htbp]
% \centerline{\includegraphics[width=0.55\linewidth]{ss3.png}}
% \caption{A multi-target scenario with the red line representing the real target trajectories and the blue representing the estimated  trajectories.} \label{fig12}
% \end{figure}

\begin{figure*}[htbp]
\vspace{-2mm}
\centerline{\includegraphics[width=0.8\linewidth]{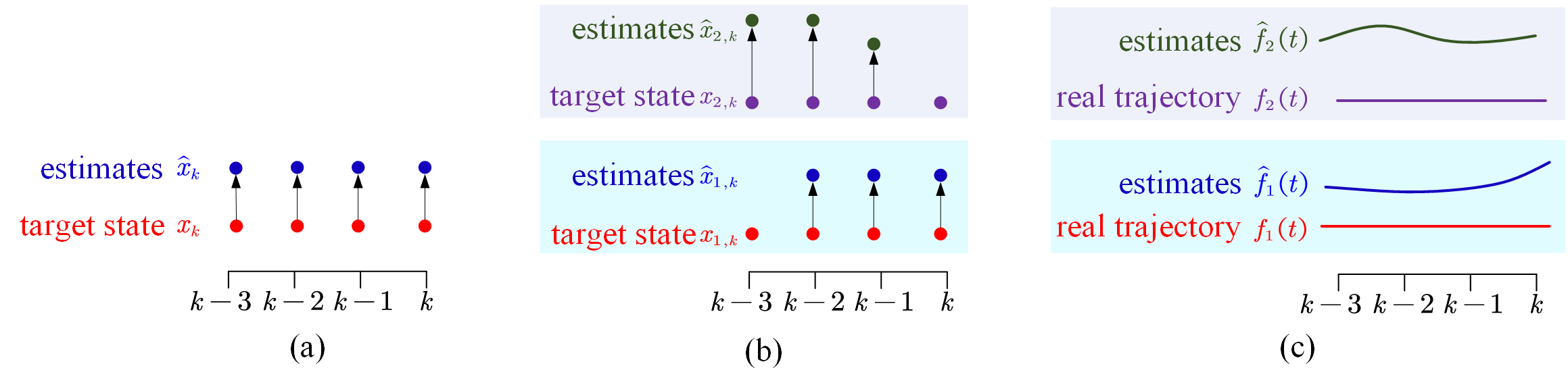}}
\vspace{-2mm}
\caption{Different tracking scenarios: (a) single target point-state estimation; (b) multiple target point-state estimation; (c) multi-target T-FoT estimation. } \label{fig:diff_track}
\end{figure*}

\subsection{Contribution and Organization of This Paper}
% While the T-FoT \eqref{eq:T-FoT} and the TSP \eqref{eq:TSP} %is quite general as it 
% can be used to describe any spatio-temporal trajectory, it is mainly suitable for a class of targets with smooth trajectory such as aircraft/ship/missile. 
% In fact, since the practical useful data regression model has to be a simple one, the T-FoT approach is usually performed in a short sliding time window \cite{li2018joint,li2023target} to model the trajectory segment of relatively low complexity. More importantly, by working with a sliding time window, the batch curve fitting is carried out in a \textit{real-time} mode, that is, the T-FoT will be updated (and so the newest state estimate will be made) as soon as new observation is available. On this point, it is analogous somehow to the moving horizon estimator \cite{Allan2019MHE}, although the output is totally different. 
% Nevertheless, it remains open how to convincingly evaluate and compare %Given a number of 
% different T-FoT estimates. %, %$\hat{f}_1(t),...,\hat{f}_m(t)$, 
% %obtained by different trackers. 
From both theoretical and practical perspectives, a metric capable of appropriately measuring the similarity of continuous-time T-FoTs that can be arbitrary forms is of vital importance for analyzing and comparing the performance of these trackers and for evaluating different strategies for algorithm improvement. 
In addition, it is required for multi-sensor cooperation and fusion in order to determine the correspondence between trajectories obtained by different sensors. 
%In addition, it is needed to determine the correspondence between the estimated and real trajectories for track management %, due to the uncertainty of the measurement-to-target association, 
%when there are multiple targets in the scenario. %In this case, the metric needs to accommodate the data association uncertainty, seeking the most suitable choice. % possibility. %In short, s
Such a metric that is mathematically rigor and meets the need of multiple spatio-temporal T-FoT evaluation constitutes the contribution of this paper. %, which also lays a foundation for . 

%This is the first part of two companion papers on the T-FoT approach. The second part addresses the polynomial T-FoT optimization with regularization \cite{Li25TFoT-part2}. 

The remainder of this paper is organized as follows. Existing representative MTT metrics are %and the T-FoT approach are 
briefly reviewed in Section \ref{sec:pre}. 
Section \ref{sec:elemental} discusses the elemental distance/metric 
to be utilized in the proposed T-FoT metric, referred to as the spatio-temporal-aligned trajectory integral distance (Star-ID), which is given later in Section \ref{sec:Star-ID}. %  metric for trajectories. 
%The main part of this paper is Section \ref{4} where the Star-ID metric is developed. Additionally, it provides  details about its parameters. 
Scenario study is given in Section \ref{sec:simula} before the paper is concluded in Section \ref{sec:con}. % and mathematical details are given in the Appendix. 
The main %mathematical 
notations used in this paper are given in Table \ref{tbl:symbols}.

\begin{table}[htbp]
\centering
\caption{List of key notations}
\begin{tabular}{ll}
\toprule
Notation & Interpretation \\
\midrule
$t$ & $t\in \mathbb{R}^+$, the continuous time (positive real number) \\
$\mathcal{F} $  & stochastic process of the state\\
$\mathbf{x}_k$ & the target state at time $k$ \\
$k$ & $k\in \mathbb{N}^+$, the discrete sampling time (integer) \\
$\mathbf{y}_k$ & the measurement at time $k$ \\
$h_k(\cdot)$ & the measurement function at time $k$ \\
$\mathbf{v}_k$ & the measurement noise at time $k$ \\
$f(t), g(t)$ & the (real or estimated) target T-FoTs \\
$F\left ( t;\mathbf{C} \right )$ &  a polynomial T-FoT with parameters $\mathbf{C}$   \\
$\mathbf{C}$ &  ste of polynomial parameters $\mathbf{C} = \left \{ \mathbf{c}_{0},\mathbf{c}_{1},\dots ,\mathbf{c}_{\gamma } \right \} $   \\
$\gamma$ &  the order of the T-FoT function \\
$T_\text{w}$  & the maximum length of the sliding time-window \\
% $C=\left \{ c_{0}  ,c_{1},\dots ,c_{\gamma } \right \} $ \\
$r$ & the trajectory dimension\\
$c$ & the cutoff coefficient for OSPA and OSPA$^{(2)}$\\
$p$ & the metric order for OSPA, OSPA$^{(2)}$ and Star-ID \\
$c_{\text{SFA}},c_{\text{SMD}}$ &segment cutoff coefficient for  FA and MD, respectively \\
%$T_{\text{SFA}},T_{\text{SMD}} $ &the duration of FA and MD, respectively \\
$T_{\text{SFA}}, T_{\text{SMD}} $ &duration of the segment FA and MD, respectively \\
$c_{\text{TFA}},c_{\text{TMD}}$ &trajectory cutoff coefficient for FA and MD, respectively \\
%$T_{\text{SFA}},T_{\text{SMD}} $ &the duration of FA and MD, respectively \\
$T_{\text{TFA}},T_{\text{TMD}} $ &duration of the trajectory FA and MD, respectively\\
\bottomrule
\end{tabular}
\label{tbl:symbols}
\end{table}

\section{Preliminaries} \label{sec:pre}
% The classic approach to target tracking since the milestone Kalman filter \cite{kalman60} is utilizing a dynamic (e.g., Markov-type) model to describe the movement of the target and a measurement model to relate the measurement of the sensor(s) with the state of the target. In addition, appreciate models are needed to characterize the false and missing data \cite{Bar-Shalom01,Sarkka13book}. All of these render the optimal state estimation in the minimum mean square error or Bayesian sense. In this case, the metric needed for estimate evaluation is applied on the point set, either regarding a single target or multiple targets as illustrated in Fig. \ref{fig:diff_track} (a) and (b), respectively.  

%Nevertheless, a general metric for multiple T-FoTs is still missing and is the goal of this paper. 
In this section, we will briefly review the OSPA-type metrics for point sets %for multiple target state estimate evaluation %, which inspired our proposed metric 
and our previous related work. 
%the concerning non-Bayesian/Kalman T-FoT approach to %both single and multiple 
%target tracking. % that is fundamentally different from the Markov-Bayes filters. 

\subsection{OSPA: Metric on Sets of Points}
The widely used OSPA distance \cite{schuhmacher2008consistent} measures the distance between two sets of points which adds consideration of the cardinality inconsistency to the Wasserstein
metric. % extends  to the general case that two sets may have different cardinalities. % and has been widely adopted to evaluate multi-target filtering performance. %Most commonly, this method is applied by computing the distance between two multi-target states, usually defined for a set of actual target states and a set of estimated target states. In this subsection, we will provide a concise overview of the OSPA metric.
%Let $d_{p}^{\left( c \right)}\left( \Phi ,\Psi \right)$ be the OSPA distance between with order $p$ and cutoff $c$. 
For $\Phi =\left\{ \phi ^{\left( 1 \right)},\phi ^{\left( 2 \right)},\cdots ,\phi ^{\left( m \right)} \right\} $ and $\Psi =\left\{ \psi ^{\left( 1 \right)},\psi ^{\left( 2 \right)},\cdots ,\psi ^{\left( n \right)} \right\} $, with $m\le n$, the OSPA distance with order $p$ and cutoff coefficient $c$ is defined as follows
\begin{align}\label{eq:ospa}
  &\text{OSPA}_{p}^{(c)}(\Phi ,\Psi) \nonumber \\
  &={\left( \frac{1}{n}\bigg( \underset{\pi \in {{\prod }_{n}}}{\mathop{\min }}\,\sum\limits_{i=1}^{m}{{\check{d}^{(c)}}{{({{\phi}^{(i)}},{{\psi}^{(\pi (i))}})}^{p}}}+{{c}^{p}}(n-m) \bigg) \right)}^{\frac{1}{p}\;},
\end{align}
where ${\check{d}^{(c)}}({\phi}^{(i)},{\psi}^{(i)}) \triangleq \min \left(c,d({\phi}^{(i)},{\psi}^{(i)})\right)$, in which $d(\cdot,\cdot)$ is a metric on the single-target state space. If $m> n$, then $d_{p}^{(c)}(\Phi ,\Psi)= d_{p}^{(c)}(\Psi ,\Phi)$.

\subsection{Unnormalized OSPA and GOSPA}
As shown, the OSPA is a normalized metric %yielding values in $[0, c]$ 
which, however, suffers from an obvious drawback as mentioned by the creators \cite{schuhmacher2008consistent} that it is insensitive to the case where one is empty. In other words, when one set is empty, the OSPA distance takes on the cutoff value of $c$ regardless of the cardinality 
of the other set. 
Another issue with the OSPA metric as pointed out by \cite{rahmathullah2017generalized} is
that it gives the same result $c$ for the distance between two sets
$\Phi$ and $\Psi$ where $|\Psi| = |\Phi| - 1$, and between two sets $\Phi$ and
$\Psi \cup  {z}$ where $d(\mathbf{x}, \mathbf{z}) \geq  c$ for $\mathbf{z} \not\in \Psi$ and $\forall \mathbf{x} \in \Phi$. %The cause is that the penalty for each unassignable point and the cut-off for a very large distance between two assigned points in the OSPA metric are the same, c
%This will lead to an unreasonable result: when the ground truth has no targets, all estimators except that yields an empty set have the same error $c$, no matter how many estimates they produce. 
Yet, this can be simply resolved by the removal of the normalization, the multiplier $\frac{1}{n}$, resulting in the unnormalized OSPA, which is proportional to the COSPA given in \cite{vu2020complete}. By further adding parameter $\alpha$ as a denominator in the cardinality error, one gets the GOSPA \cite{rahmathullah2017generalized, Angel2020GOSPA}, c.f., \eqref{eq:ospa} 
\begin{align}\label{eq:Gospa}
  &\text{GOSPA}_{p}^{(c)}(\Phi ,\Psi) \nonumber \\
  &={\left(  \underset{\pi \in {{\prod }_{n}}}{\mathop{\min }}\,\sum\limits_{i=1}^{m}{{\check{d}^{(c)}}{{({{\phi}^{(i)}},{{\psi}^{(\pi (i))}})}^{p}}}+{ \frac{{c}^{p}}{\alpha} }(n-m) \right) }^{\frac{1}{p}\;}.
\end{align} 

%introducing an additional parameter $\alpha$ to control the cardinality mismatch penalty
Differently from the OSPA that is a per-target error, the GOSPA is an unnormalized metric which sums up all distances between the finite sets and increases with the increase of the sizes of the sets. Sensitivity of the OSPA and GOSPA metrics to the parameters has been analyzed in \cite{Barrios23comparison}. 

\subsection{$\text{OSPA}^{(2)}$: Metric on Sets of Tracks}
The so-called $\text{OSPA}^{(2)}$ metric \cite{beard2017ospa,Beard20LMB} measures the distance between two sets of tracks which does not only capture target state errors, cardinality errors, but also phenomena such as track switching and track fragmentation. Here, each track is a time-series of point-estimates as shown in Fig. \ref{fig:diff_track} (b). % where each color denotes a track. 

 % Let $X=\left\{\Phi_{S^{\left( 1 \right)}}^{(1)}, \Phi_{S^{\left( 2 \right)}}^{(2)},\cdots , \Phi_{S^{\left( m \right)}}^{(m)}\right\}$ and $Y=\left\{\Psi_{T^{\left( 1 \right)}}^{(1)}, \Psi_{T^{\left( 2 \right)}}^{(2)},\cdots , \Psi_{T^{\left( n\right)}}^{(n)}\right\}$ be two sets of tracks, where each $\Phi_T^{(i)}$/$\Psi_T^{(j)}$ is a set of points with domain $T$, $1 \leq i\leq m, 1 \leq j\leq n $. The domain $T$ of a track $\Phi_T/\Psi_T$ represents the set of all times at which the object exists, namely all time $t \in T$. Assuming $ m \le n$, the $\text{OSPA}^{(2)}$ distance ${d}_{p,q}^{\left( c,w \right)}\left( X,Y\right)$ between $X$ and $Y$ is defined as

 Consider two set of tracks $\mathbf{X}=\left\{\Phi^{(1)}, \Phi^{(2)},\cdots , \Phi^{(m)}\right\}$ and $\mathbf{Y}=\left\{\Psi^{(1)}, \Psi^{(2)},\cdots , \Psi^{(n)}\right\}$, % be two sets of tracks, 
 where each track $\Phi^{(i)}$ or $\Psi^{(j)}$ consists of a time-series of points, $1 \leq i\leq m, 1 \leq j\leq n $. %The domain $T$ of a track $\Phi_T/\Psi_T$ represents the set of all times at which the object exists, namely all time $t \in T$. 
 Assuming $ m \le n$, the $\text{OSPA}^{(2)}$ distance ${d}_{p,q}^{\left( c,w \right)}\left( \mathbf{X},\mathbf{Y}\right)$ between $\mathbf{X}$ and $\mathbf{Y}$ is defined as
\begin{align}\label{eq:ospa2}
  &\text{OSPA}_{p,q}^{(2)\left( c,w \right)}\left( \mathbf{X},\mathbf{Y} \right)  \nonumber \\
  &={\left(\frac{1}{n}\bigg( \underset{\pi \in {{\prod }_{n}}}{\mathop{\min }}\,\sum\limits_{i=1}^{m}{{{\widetilde{d}_{q}^{\left( c,w \right)}\left( \Phi^{(i)},\Psi^{(\pi(i))}\right)}^{p}}}+{{c}^{p}}(n-m) \bigg) \right)}^{\frac{1}{p}}
\end{align}
where %$c$ is the cutoff coefficient, $p$ is the order, 
$q$ is the order of the base distance, all positive-defined %and $w$ is a collection of convex weights, %Note that for clarity of notation, we have omitted the demain of each track, using the abbreviations $\Phi^{\left( i \right)}=\Phi_{S^{\left( i \right)}}^{\left( i \right)}$ and $\Psi^{\left( i \right)}=\Psi_{T^{\left( i \right)}}^{\left( i \right)}$.  Let  $f_T$ be a function with domain $T$. 
and 
$\widetilde{d}_{q}^{\left( c,w \right)}\left( \Phi,\Psi \right)$ between two tracks $\Phi,\Psi $ is defined as 
\begin{equation} \label{eq:ospa2-trackDis}
\widetilde{d}_{q}^{\left( c,w\right)}\left( \Phi,\Psi\right) =\left( \sum_{t=1}^{K}{\left( w\left( t \right) d^{\left( c \right)}\left( \left\{ \Phi\left( t \right) \right\} ,\left\{ \Psi\left( t \right) \right\} \right) \right) ^q} \right) ^{\frac{1}{q}}
\end{equation}
where $w(t)>0$ is a positive weighting function defined for $t\in \left\{ 1,\cdots , K \right\} $ which includes all time indices from the
beginning to the end of the tracking scenario, such that $\sum_{t=1}^{K}{w\left( t \right)}=1$ and the singleton-based OSPA distance reduces to the following 
\begin{equation}
d^{(c)}\left( \phi, \psi\right) = \left\{ \begin{array}{l}
\begin{aligned}
	0,~~ &|\phi| = |\psi| =0\\
        c, ~~ &|\phi| \neq |\psi| \\
	\min \left(c,d({\phi} , {\psi} )\right), ~ & |\phi| = |\psi| =1\\
 \end{aligned}
\end{array}  \right. 
\end{equation}

Note that ${d}_{p,q}^{\left( c,w \right)}\left( \mathbf{X},\mathbf{Y} \right)  = {d}_{p,q}^{\left( c,w \right)}\left( \mathbf{Y},\mathbf{X} \right)  $ if $m>n$. While the OSPA distance indicates the per-target error, the $\text{OSPA}^{(2)}$ distance can be interpreted as the time-averaged per-track error. In both, smaller values indicate higher accuracy. %better filtering/tracking performance.

\subsection{Integral Multi-target Trajectory Assignment (IMTA) Metric}

An intuitive approach to the T-FoT evaluation is discretizing the continuous-time curve function into discrete points and using established metrics such as the $\text{OSPA}^{(2)}$ for evaluation. This, however, will cause the loss of detailed information of the trajectory, such as the length and smoothness, and loss of the advantage of obtaining the whole trajectory. For example, this approach can not distinguish the two different trajectory estimates shown in Fig. \ref{fig:trajec_discret} and therefore undermines the purpose of performance evaluation metric. Instead, we need to directly calculate the distance between two sets of continuous-time functions. %brings in a number of approximation  %However, it became apparent that this approximation has inherent limitations and shortcomings. To name a few, 

%highlighting the need for a more advanced methodology.
%Thus, it is crucial to propose a metric for continuous-time trajectory evaluation.
\begin{figure}[htbp]
\centerline{\includegraphics[width=0.6\linewidth]{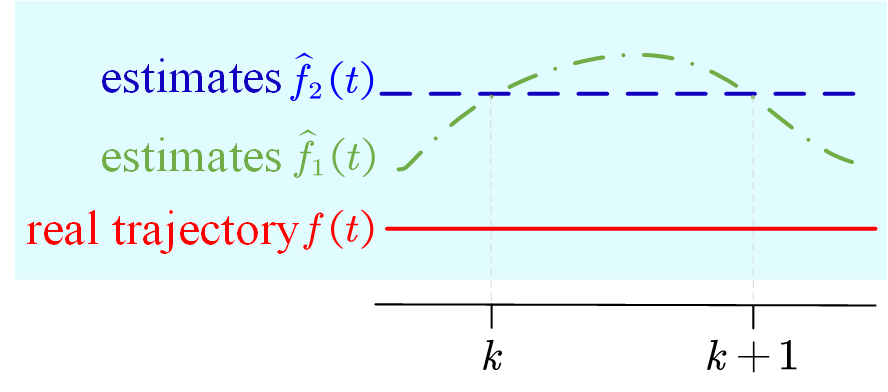}}
\vspace{-3mm}
\caption{Two different T-FoTs share the same discretizing points. %Estimated set of trajectories  against the ground truth, where the solid line represents the true trajectory, the dashed line and the dashdotted line represents the estimates.
} \label{fig:trajec_discret}
\end{figure}

Our earlier attempt as called the IMTA metric in \cite{xin2022metric} is first reviewed.  
Let $ \mathcal{F} =\{{{f}_{1}},{{f}_{2}},\ldots ,{{f}_{ m}}\}$ and $ \mathcal{G} =\{{{g}_{1}},{{g}_{2}},\ldots ,{{g}_{n}}\}$ be finite subsets of continuous-time T-FoTs of size $m$ and $n$, respectively, one being the ground truth and the other estimates.
% Let ${{\Pi }_{z}}$ be the set of all arbitrary combination pairing (not limited to one-to-one pairing) of $\{1,2,\cdots ,z\}$ for any $z\subseteq \left [ 1, \max\left ( m,n \right )  \right ] $ and any element $\pi \in {{\Pi }_{z}}$ be a sequence $\{\pi (1),\pi (2),\cdots ,\pi (z)\}$. 
If $m \le n$, the IMTA metric is defined as
\begin{align}
\text{IMTA}_{r}^{(c)}(\mathcal{F},\mathcal{G})
& = \underset{\pi \in \Pi_n}
{\mathop{\min }}
\bigg(  \sum\limits_{i=1}^{m} d_{r}^{(2)}\left({{f}_{i}},{{g}_{\pi (i)}}\right) +({{c}_\text{TFA}} {{T}_\text{TFA}}\nonumber \\
&~~~~~~+{{c}_\text{TMD}}  {{{T}}_\text{TMD}})^{ r }
\bigg), \label{eq:IMTA}
\end{align}
where %$c_\text{TFA}$ determines the weighting of how the metric penalizes FA errors caused by the estimated trajectories that have no matched true trajectory, while $c_\text{TMD}$ is the weighting of MD errors caused by the ground truth which do not assign any estimated trajectory. $T_\text{TFA}$ and $T_\text{TMD}$ represent the length of time of the false trajectories and undetected real trajectories, respectively. 
$T_\text{TFA}$ and $T_\text{TMD}$ denote the duration time of the false and missed trajectories, and $c_\text{TFA}$ and $c_\text{TMD}$ are the corresponding penalty in two cases, respectively,
and the pairwise distance \( d_{r}^{(2)}\left({f},{g}\right) \) between two T-FoTs $f, g$ is as follows
\begin{equation} \label{eq:base-distance-IMTA}
d_{r}^{(2)}(f,g)=\frac{1}{\tau}\Big(d_{t_1,t_2}^{(2)}(f,g)+ ({c_\text{SFA}}  {T_\text{SFA}} + {c_\text{SMD}}  {T_\text{SMD}})^{r}\Big),
\end{equation}
where $t_{1}, t_{2}$ indicate the starting and ending times of the common time domain of the two T-FoTs $f, g$, respectively, $\tau={({t_2} - {t_1} + {T_\text{SFA}} + {T_\text{SMD}}) ^{ r}}$, %$f(t)$ and $g(t)$ indicate the overlapping partial function curves of the true trajectory $x$ and the estimated trajectory $y$ in the time domain, respectively, %. The matching period trajectory starts at time $t_{1}$ and ends at time $t_{2}$, while 
$T_\text{SFA}$ and $T_\text{SMD}$ denote the duration time of the FA and MD, and $c_\text{SFA}$ and $c_\text{SMD}$ are the corresponding FA and MD penalty, respectively, %$\left| I \right|$ indicates the magnitude of the dimension in the position space. 
$d_{t_1,t_2}^{(2)}(f,g)$ is the $\ell_2$ distance between two functions \(f\) and \(g\) in the common time domain $[t_1, t_2]$, c.f. \eqref{eq:l_p}
\begin{equation} \label{eq:L2-integral}
d_{t_1,t_2}^{(2)}(f,g) =\int_{t_1}^{t_2}{\left( f\left( t \right) -g\left( t \right) \right) ^\mathrm{T}\left( f\left( t \right) -g\left( t \right) \right) dt}.     
\end{equation}

%
%
%In Section \ref{3}, we will provide a detailed explanation of the various types of errors that need to be considered in the design of the metric. This will include consistent (localization) error, inconsistent (cardinality) error, and shape (smoothness) error, among others. We will discuss the impact of these errors on the performance evaluation of target tracking systems and propose solutions to address them.

As shown above, the pairwise distance between associated trajectories is defined by the integral of the divergence of two functions, whether $\ell_1$ or $\ell_2$ norms, over time. The unit of the metric is spatio-temporal such as $\text{m}\cdot\text{s}$. This is significantly different from the metric for point sets. %Yet, no metric is perfect. 
The IMTA metric, however, has two limitations and drawbacks: 
\begin{itemize}
    % \item IMTA ignores the shape of the trajectory that reflects the meticulous change of the trajectory over time, an important aspect of the similarity. By merely calculating the integral-based distance as in \eqref{eq:base-distance-IMTA}, it may not be able to distinguish trajectory estimates with different lengths as illustrated in Fig. \ref{fig:smoothness}.
    \item IMTA does not properly coordinate the localization error and FA/MD error since in the former the error in each dimension is summarized while in the latter it is accounted for by the exponential calculation. Two parts have different units of error unless $r=1$. This theoretical drawback can be fixed by substituting the exponential calculation in the latter with multiplication. 
\end{itemize}

% \begin{figure}[htbp]
% \centerline{\includegraphics[width=0.4\columnwidth]{smooth_scene.png}}
% \caption{Two T-FoT estimates share different smoothness but the same pairwise distance as defined in \eqref{eq:base-distance-IMTA} from the real trajectory. % he true trajectory (the red line) and two estimated trajectories (the blue dashed line and the dark blue dashdotted line).
% }
% \label{fig:smoothness}
% \end{figure}

We next propose a new spatio-temporal metric that overcomes the above limitations and drawbacks. %takes into account the smoothness of the trajectory and properly balance the localization with the FA/MD errors. 
%a comprehensive assessment of smoothness in various shapes and avoids insensitivity to trajectory shape caused by integral calculation.

\section{Elemental Segment Distance} \label{sec:elemental}

% \begin{figure*}[htbp]
% \centerline{\includegraphics[width=0.75\textwidth,height=0.3\textwidth]{liucheng.png}}
% \caption{Flowchart of various errors contained in continuous-time trajectory function metric.}
% \label{fig:diff_time}
% \end{figure*}

%We focus on the similarity of two polynomial functions.
%%Similar to the case of discrete target states, the measurements of the continuous time track also have false  or missed trajectories. However, the difference lies in the fact that these issues are represented in the form of a curve, requiring a unified approach to time consistency.
%% As shown in Fig. \ref{fig1},
%In general, the estimated trajectory and the real trajectory  do not  match in the time domain, leading to the trajectory consistency and inconsistency in time, respectively. Furthermore, the time consistency component consists of two main factors: the distance difference and the shape difference. Fig. \ref{fig:diff_time}  illustrates the specific flowchart for the continuous-time trajectory metric.
%%\begin{figure}
%%\centerline{\includegraphics[width=\columnwidth]{fig1.png}}
%%\caption{Evaluate the more general case of two continuous function trajectories.}
%%\label{fig1}
%%\end{figure}

%The flowchart in Fig. \ref{fig:diff_time} illustrates the specific process for evaluating the continuous-time trajectory metric, which focuses on comparing the similarity between two continuous functions.

Similar to the distance between point sets, the distance between T-FoT sets that needs to solve the trajectory-to-trajectory (T2T) association problem first. But differently, the T-FoT association needs to be carried out in the continuous time domain and both the trajectory FA/MD (TFA/TMD) and segment FA/MD (SFA/SMD) are defined too in the continuous time domain. More formally speaking, %, However, unlike the OSPA metric, which primarily incorporates the cardinality error (discrepancy in the number of elements between two sets) and the positioning error (state distance between elements within the sets) of finite sets, defining the distance between trajectory sets requires consideration of the following two crucial aspects:

\begin{itemize}
\item \textit{Trajectory temporal-alignment}: The T-FoT divergence needs to account for the temporal overlap and non-overlap parts, separately. %They are referred to as \textit{consistency} and \textit{non-consistency}, respectively. 
The former constitutes the spatial difference between the temporally aligned T-FoT segments while the latter constitutes the SFA or SMD, namely unaligned segment error. %SFA and SMD exist due to the fact that two T-FoT estimates of the same target may not comply with each other over all time. 
As a result, two T-FoTs that overlap spatially but not temporally are different as illustrated in Fig. \ref{fig:diff_time}. 
%If their corresponding time instances differ, the overlapping portion becomes irrelevant, and there is no need to calculate the distance. Consequently, the computation of continuous-time domain track distance requires reliable time integration as the foundation, as shown in Fig. \ref{fig:diff_time}.
\item \textit{FA and MD errors}: The FA and MD are not solely determined by a simple disparity in the number of tracks (i.e., TFA/TMD) but also by the trajectory non-overlap/unaligned segments (i.e., SFA/SMD). In addition, both cases need to take into account the duration of the concerning FA/MD. For instance, the cost of missing a track (or a segment) spanning $1$s is evidently different from that spanning $5$s \footnote{Similar issue has been overlooked in existing track metrics such as the $\text{OSPA}^{(2)}$ distance in which, as shown in \eqref{eq:ospa2}, each FA or MD is recognized the same as error $c$ in the result, no matter how long each FA or MD exists. % in the tracking.
}. In an extreme case when the duration of the MD is so short that it is less than a sampling period as shown in Fig. \ref{fig:MD<1}, no existing point-estimate metric can properly calculate the MD error. In the T-FoT approach, the penalty for FA and MD needs to take into account the duration. 
\end{itemize}
%  as shown in Fig. \ref{fig:diff_time}.

\begin{figure}[htbp]
\centerline{\includegraphics[width=\columnwidth]{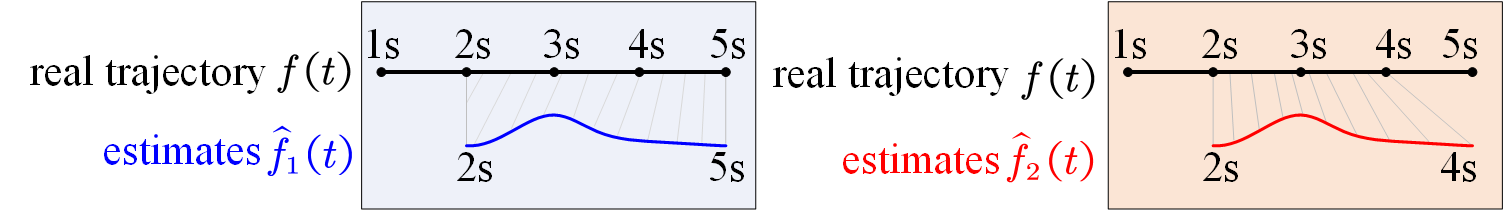}}
\caption{Two T-FoT estimates overlap spatially but not temporally.}
\label{fig:diff_time}
\end{figure}

 \begin{figure}[htbp]
\centerline{\includegraphics[width=.3\columnwidth]{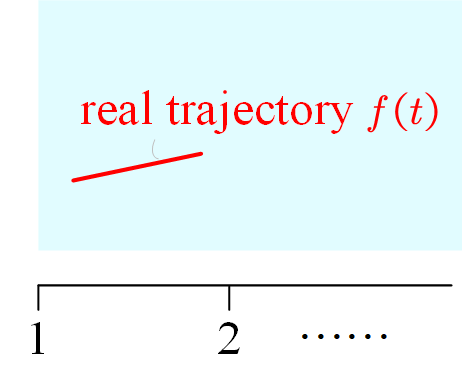}}
\caption{A target exists less than a sampling period and is undetected by the estimator.}
\label{fig:MD<1}
\end{figure}

% The arbitrary finite subsets of continuous-time trajectory curve function $f=\left\{ f^{\left( 1 \right)}\left( t \right) ,f^{\left( 2 \right)}\left( t \right) ,\cdots ,h^{\left( \text{d} \right)}\left( t \right) \right\}^T$ and $G=\left\{ g^{\left( 1 \right)}\left( t \right) ,g^{\left( 2 \right)}\left( t \right) ,\cdots ,g^{\left( \text{d} \right)}\left( t \right) \right\}^T$, where $h^{(i)}\left( t \right) $ represents the continuous-time function of the $i$-th dimension. 

Therefore, the first step to calculate the distance between two trajectories is to divide them into segments in the time domain according to their temporal alignment. The second step is to calculate the spatio-temporal-domain-based distance with regard to each segment, which forms the elemental units to construct our final metric for T-FoT comparison.  

\subsection{Temporally Aligned Segment}

%The measurement of the distance between the estimated and real trajectory is crucial for evaluating the performance of tracking algorithms. A widely used method is the \(\ell_1\)-norm \cite{qiu2014line,jianbin2006new}, which is defined as the integral of the Euclidean distance between two functions over their common time domain.
%  Given two  continuous  functions \(f(t)\), \(g(t)\) defined on the interval \([t_1,t_2]\), the \(L_1\) distance  between them can be expressed as

%The \(\ell_1\)-norm \cite{qiu2014line,jianbin2006new} is particularly useful for comparing functions with similar shapes. It is robust to outliers and noise in the data and gives more weight to the larger differences between the two functions \cite{bastian2023comparing}. By calculating the absolute value of the difference, we eliminate the possibility of negative values, ensuring that larger differences have a more significant impact on the overall distance. 

We define the spatio-temporal distance between two temporally aligned segments as follows.

\begin{definition}[\(\ell_p\) distance]
The \(\ell_p\) distance of two T-FoTs \(f\), \(g\) in the same time interval \([t_1,t_2]\) is defined as
\begin{equation} \label{eq:l_p}
d_{t_1,t_2}^{(p)}\left( f,g\right)  \triangleq \int_{t_1}^{t_2}{ \| f(t)-g(t)\|_pdt},
\end{equation}
where the $\ell_p$ norm is defined as $\lVert \mathbf{x} \rVert_p \triangleq \left(\sum_{i=1}^r{\left|\mathbf{x}[i]\right|^p}\right)^{\frac{1}{p}}$. % with $\mathbf{x}=[\mathbf{x}[1],\mathbf{x}[2],...,\mathbf{x}[r]]^\text{T}$.
\end{definition}

\subsection{Temporally Unaligned Segment: SFA/SMD}
%The imperfections in the sensor may cause FAs or MDs, leading to inconsistency within the time domain and affecting the accuracy of target tracking. Measuring these inconsistencies is a challenging problem that requires careful consideration.

%\begin{definition}[SFA]
SFA refers to the estimated trajectory segment that is not assigned to any real trajectory. The SFA duration corresponding to the time-window \([t_{\text{SFA},1},t_{\text{SFA},2}]\) is given by 
\begin{equation}
 T_{\text{SFA}}=t_{\text{SFA},2}-t_{\text{SFA},1}.
 \end{equation}
%\end{definition}

%\begin{definition}[SMD]
SMD refers to the real trajectory segment that is not assigned to any estimated trajectory. The SMD duration corresponding to the time-window \([t_{\text{SMD},1}, t_{\text{SMD},2}]\) is given by
\begin{equation}
T_{\text{SMD}}=t_{\text{SMD},2}-t_{\text{SMD},1}.
\end{equation}
%\end{definition}

Correspondingly, penalty should be applied to the SFA and SMD by using the segment cutoff coefficients $c_{\text{SFA}}$ and $c_{\text{SMD}}$, respectively. Similarly, we denote the duration and cutoff coefficients for the TFA and TMD by $T_{\text{SFA}}, c_{\text{TFA}}$ and $T_{\text{SMD}}, c_{\text{TMD}}$, respectively. %We omit the detail here but note that typically, $c_{\text{TFA}} \geq c_{\text{SFA}}, c_{\text{TMD}}\geq c_{\text{SMD}}$.  
% To quantify the errors resulting from inconsistent trajectories, penalty factors are introduced. These factors can be expressed as follows
% \begin{equation}
%  e _{incon} =c_{\text{SFA}} T_{\text{SFA}}+c_{\text{SMD}}  T_{\text{SMD}},
% \end{equation}
% where $\boldsymbol{c}=\left\{ c_{\text{SFA}},c_{\text{SMD}} \right\} $ is the penalty coefficients for FA and MD errors, which is determined based on the specific requirements of the scene. 
A longer duration and a higher penalty will indicate a greater distance between the estimated and real trajectories. Note that these cutoff parameters should not be arbitrarily chosen but follow certain rules as to be addressed in section \ref{sec:Star-ID}.B. % in order to ensure that the selection is realistic and appropriate. 

%Based on the analysis of the aforementioned departmental errors, 
\subsection{Distance between Two T-FoTs}
A convincing distance between two sets of T-FoTs has to be built on the base of a meaningful base-distance between two single T-FoTs. 
After temporally aligning the estimated and actual T-FoTs $f, g$, the spatio-temporal distance can be defined as follows:
\begin{align} 
d_{\boldsymbol{c}}^{(p)}\left( f,g \right) = \Big( r \left( c_{\text{SFA}} T_{\text{SFA}}+c_{\text{SMD}} T_{\text{SMD}} \right) ^p +d_{t_1,t_2}^{(\boldsymbol{c},p)} (f,g) \Big)^{{1}/{p}\;}, \label{eq:f-g-dist}
\end{align}
where 
\begin{align}\label{dcpt1t2}
d_{t_1,t_2}^{(\boldsymbol{c},p)} (f,g) 
\triangleq \min \Big(&\big(d_{t_1,t_2}^{(p)} (f,g)\big)^p , \nonumber \\ 
&~~~ r(c_{\text{SFA}}+c_{\text{SMD}})^p(t_2-t_1)^p  \Big).
\end{align}
% \begin{align}\label{dcpt1t2}
% d_{t_1,t_2}^{(\boldsymbol{c},p)} (f,g) 
% \triangleq \min \Big(&\big(d_{t_1,t_2}^{(p)} (f,g)\big)^p +\big(d_{t_1,t_2}(f,g)\big)^p , \nonumber \\ 
% &r(c_{\text{SFA}}+c_{\text{SMD}})^p(t_2-t_1)^p  \Big).
% \end{align}
Here, \([t_1,t_2]\) denotes the aligned time interval of \(f\) and \(g\), $p$ represents the metric order s.t. $1\le p< \infty $,
%The T-FoT approach directly estimates the track/trajectory, which includes the position of the target as well as its velocities and accelerations (corresponding to the first and second derivatives of position with respect to time). 
the parameter \(r\) scales the distance with the dimension of the trajectory in order to comply with the $\ell_p$ norm  in \eqref{eq:l_p}. %By introducing the parameter \(r\), we can generalize the tracking problem to handle targets in different dimensions, such as two-dimensional (2D) or three-dimensional (3D) space. This allows us to apply the same target tracking algorithm in different dimensional spaces without the need for significant modifications or adjustments to the algorithm. \(r\) is an important consideration when designing performance metric, as it can significantly impact the accuracy and performance of the system.
\begin{definition}[Distance metric] \label{def:distance}
Given a metric $d(\cdot,\cdot)$ and any three items $f, g$ and $h$, $d$ is a distance metric if it has the following four properties: 
\begin{enumerate}
\item non-negativity: \( d\left( f,g \right) \ge 0\)
\item identity: \(d\left( f,g \right)= 0 \Longleftrightarrow f(t)=g(t) \)
\item symmetry: \(d\left( f,g \right)=  d\left( g,f \right) \)
\item triangle inequality: \( d\left( f,g \right) \le d\left( f,h \right) +d\left( h,g \right)\)
\end{enumerate}
\end{definition}

\begin{theorem} \label{theorem_d(f,g)}
When $c_{\text{SFA}}=c_{\text{SMD}} >0$ and $ 1\le p< \infty$, $d_{\boldsymbol{c}}^{(p)}\left( f,g \right)$ is a distance. 
\end{theorem}
\begin{proof}
Following Definition \ref{def:distance}, the proof is given in Appendix \ref{sec:proof-d(f,g)}. 
\end{proof}

\begin{remark}
When there are multiple trajectories in either the estimated or actual T-FoT set, the T2T association is needed to calculate the distance. Different association results lead to different $T_\text{TFA},T_\text{TMD},T_\text{SFA}$ and $T_\text{SMD}$. 
Therefore, we hereafter write these parameters more formally by $T^{\theta}_\text{TFA},T^{\theta}_\text{TMD},T^{\theta}_\text{SFA},T^{\theta}_\text{SMD}$, where $\theta$ denotes the global T2T association between the estimated and actual T-FoT sets. %in the case there are multiple association possibilities.  
% \end{remark}
% \begin{remark}
Once the pairwise T-FoT distance is properly defined like in \eqref{eq:f-g-dist}, one can solve the T2T association problem between two sets of T-FoTs by using methods such as the known Hungarian/Kuhn-Munkres algorithm \cite{Munkres1957}. % Kuhn1955, %To avoid distracting the attention to our key contribution, the detail of this is omitted. %, however, is not the focus of this paper. %To avoid distraction from our key contribution, we leave the T2T im
\end{remark}

\section{Proposed Metric: Star-ID} \label{sec:Star-ID}
%Based on the Section \ref{3}, the IMTA metric is highly sensitive to the shape and smoothness of the trajectories, which can lead to inconsistent results. We introduce our new performance metric, the Star-ID metric, which aims to address this issue by introducing the concept of curve length. The Star-ID metric is defined based on the IMTA metric, but with modifications that make it more robust to differences in trajectory shape and scale.

%Let $X=\{{{f}_{1}}(t),{{f}_{2}}(t),\ldots ,{{f}_{m}}(t)\}$ and $Y=\{{{g}_{1}}(t),{{g}_{2}}(t),\ldots ,{{g}_{n}}(t)\}$ be finite subsets of continuous-time trajectory curve functions.

%Elaborating on 
Integrating the pairwise distance between two T-FoTs with the FA and MD errors, we proceed to give the final multi-target T-FoT metric and interpret all user-specific parameters. 
\subsection{Main Result}
% Let $\mathcal{H}=\left\{ H_1\left( t \right) ,H_2\left( t \right) ,\cdots ,H_a\left( t \right) \right\}  $ and $\mathcal{G}=\left\{ G_1\left( t \right) ,G_2\left( t \right) ,\cdots ,G_b\left( t \right) \right\} $ be two sets of trajectories. 
Without loss of generality, let $\mathcal{F} =\{{{f}_{1}},{{f}_{2}},\ldots ,{{f}_{ m}}\}$ and $ \mathcal{G} =\{{{g}_{1}},{{g}_{2}},\ldots ,{{g}_{n}}\}$ be the real and estimated T-FoT sets. Note that in online use of the metric, one may consider merely the part of these T-FoTs in the sliding time-window used for T-FoT estimation. Furthermore, %denote the global T2T association between two sets by $\theta$ and 
let us use $\Theta $ to represent the set of all possible T2T associations. If $f \in \mathcal{F}$ is associated with $g \in \mathcal{G}$ in the global T2T association $\theta \in \Theta$, we define that $\mathbf{1}_{\theta}\left(f, g \right) =1 $. Otherwise, $\mathbf{1}_{\theta}\left(f, g \right) =0 $.  
%and let the T-FoT association function $f= \lambda \left(g \right)$ represents the association of the T-FoT $f$ from the set $\mathcal{F}$ with $g$ from the set $\mathcal{G}$; the inverse function is also valid, $g= \lambda ^{-1}\left( f \right)$. We further use % for the trajectories. %, and if there is no association, it corresponds to an empty set. %Again, we leave the issue how to get a concise yet good set.
The proposed Star-ID metric is formally defined as follows
%\begin{definition}[Star-ID]
\begin{align}
{d}_{\boldsymbol{c}}^{(p)}( \mathcal{F},\mathcal{G} ) &=\underset{\theta \in \Theta}{\min}\bigg( \sum_{i=1}^n{
%\bar{1}_f\left( \lambda \left( g_i \right) \right) \cdot 
r\left( c_{\text{TFA}} T^{\theta}_{\text{TFA,}i} \right) ^p} %\nonumber  \\
%&~~~~
+ \sum_{j=1}^m{
%\bar{1}_g\left(\lambda ^{-1}\left( f_j \right) \right) \cdot 
r\left( c_{\text{TMD}} T^{\theta}_{\text{TMD,}j} \right) ^p} \nonumber  \\
&~~~~+\sum_{j=1}^m\sum_{i=1}^n{
%\mathbf{1}_{f_j}\left(\lambda \left( g_i \right) \right)
\mathbf{1}_{\theta}\left(f_j, g_i \right)
 \Big( \check{d}_{\boldsymbol{c}}^{(p)}\left(  f_j ,g_i \right) \Big)^p} \bigg)^{{1}/{p}\;} \label{eq:Star-ID-def} \\
% \end{align}
% where
% \begin{align}
% & \mathbf{1}_{\theta}\left(f, g \right) \triangleq \left\{ \begin{array}{l}
% 	1, g = f\\
% 	0, \text{otherwise}\\
% \end{array} \right.  
% \\
s.t. ~~ & \sum_{j=1}^m \mathbf{1}_{\theta}\left( f_j ,g_i  \right) \leq 1, \forall 1\leq i\leq n \label{eq:f-g-1}, \\
&  \sum_{i=1}^n \mathbf{1}_{\theta}\left( f_j ,g_i  \right) \leq 1 , \forall 1\leq j\leq m . \label{eq:f-g-2}
% \bar{1} _f\left( g  \right) \triangleq \left\{ \begin{array}{l}
% 	1, g \neq f\\
% 	0, \text{otherwise}\\
% \end{array} \right. .
\end{align}
where 
\begin{align}\label{d_cpfg_m}
&\check{d}_{\boldsymbol{c}}^{(p)}\left(f,g\right)\nonumber  \\
&=\min \left(d_{\boldsymbol{c}}^{(p)}\left(  f ,g \right),\big( r\left( c_{\text{TFA}} T_\text{g} \right) ^p+r\left( c_{\text{TMD}} T_\text{f} \right) ^p \big) ^{\frac{1}{p}}\right).
\end{align}
Here, \eqref{eq:f-g-1} and \eqref{eq:f-g-2} ensure injective T2T association. 
%\end{definition}

%and typically, $c_{\text{TFA}}>c_{\text{SFA}}, c_{\text{TMD}}>c_{\text{SMD}}$. 
% $c_{\text{TFA}}$ determines the weighting of FA errors, which are caused by estimated trajectories that do not correspond to any true trajectory,  while $c_{\text{TMD}}$ determines the weighting of MD errors, which arise from true trajectories that do not have any corresponding estimated trajectory. The sizes of $c_{\text{TFA}}, c_{\text{TMD}}$ should be different from the size settings of $c_{\text{SFA}}, c_{\text{SMD}}$, $c_{\text{TFA}}, c_{\text{TMD}}>c_{\text{SFA}}, c_{\text{SMD}}$. $T^{\theta}_{\text{TFA,}i}$ and $T^{\theta}_{\text{TMD,}j}$ denote the time duration for the $i$-th track being an FA and the $j$-th track being an MD, respectively.
\begin{theorem} \label{theorem_d(F,G)}
When $c_{\text{SFA}}=c_{\text{SMD}} >0, c_{\text{TFA}}=c_{\text{TMD}} >0 $ and $1\le p< \infty$, ${d}_{\boldsymbol{c}}^{(p)}( \mathcal{F},\mathcal{G} )$ is a distance. 
\end{theorem}

\begin{proof}
Similar to the proof of Theorem \ref{theorem_d(f,g)}, the proof of the four properties of ${d}_{\boldsymbol{c}}^{(p)}( \mathcal{F},\mathcal{G})$ %namely non-negativity, identity, symmetry and triangle inequality, 
is given in Appendix \ref{sec:proofStar-ID}. 
\end{proof}

% To prove that ${d}_{\boldsymbol{c}}^{(p)}( \mathcal{F},\mathcal{G} )$ is a distance metric, the following four properties must be satisfied and are sufficient:
% \begin{enumerate}
% \item non-negativity: \({d}_{\boldsymbol{c}}^{(p)}( \mathcal{F},\mathcal{G} )\ge 0\)
% \item identity: \({d}_{\boldsymbol{c}}^{(p)}( \mathcal{F},\mathcal{G} )= 0 \Longleftrightarrow \mathcal{F}=\mathcal{G} \)
% \item symmetry: \({d}_{\boldsymbol{c}}^{(p)}( \mathcal{F},\mathcal{G} )=  {d}_{\boldsymbol{c}}^{(p)}\left( \mathcal{G},\mathcal{F} \right) \)
% \item triangle inequality: \({d}_{\boldsymbol{c}}^{(p)}(\mathcal{F},\mathcal{G})\le {d}_{\boldsymbol{c}}^{(p)}( \mathcal{F},\mathcal{H})+{d}_{\boldsymbol{c}}^{(p)}(\mathcal{H},\mathcal{G} )\)
% \end{enumerate}

\begin{figure}[htbp]
\centerline{\includegraphics[width=0.9\linewidth]{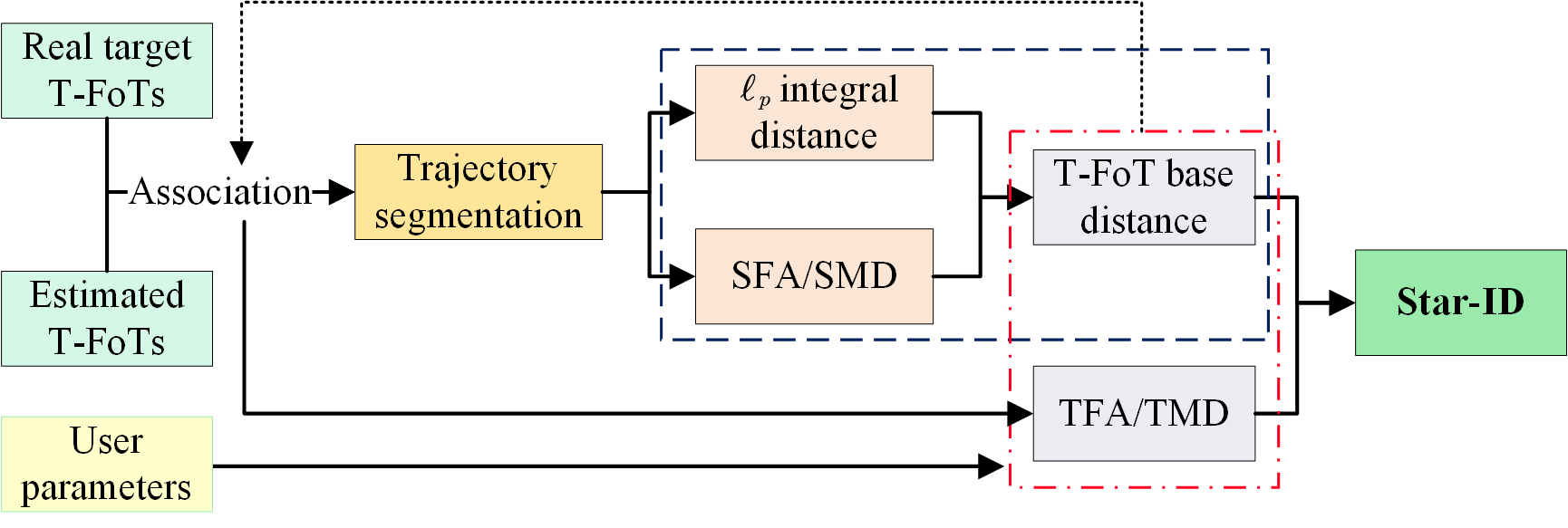}}
\caption{Flowchart for calculating the Star-ID metric.}
\end{figure}

\begin{remark}
    $c_{\text{SFA}}=c_{\text{SMD}}>0$ and $c_{\text{TFA}}=c_{\text{TMD}}>0$ are necessary for the Star-ID to be a distance. However, in practice one may use different penalties for FA and MD although this will go against with the definition of a distance.  
\end{remark}

As shown in \eqref{eq:Star-ID-def}, the Star-ID calculates the distance between two T-FoT sets, summing up the distances between all associated trajectories and the FA and MD errors over all time. Based on this measure, we further define the time-averaged Star-ID (of spatial measure unit) %, Star-ID per-target and time-averaged Star-ID per-target respectively 
as follows
\begin{definition}[Time-averaged Star-ID]
\begin{equation}
    \bar{d}_{\boldsymbol{c}}^{(p)}( \mathcal{F},\mathcal{G} )  \triangleq \frac{1}{k-k'} {d}_{\boldsymbol{c}}^{(p)}( \mathcal{F},\mathcal{G} ) 
\end{equation}
% \begin{definition}[Star-ID per-target]
% \begin{equation}
%     \bar{d}_{\boldsymbol{c},2}^{(p)}( \mathcal{F},\mathcal{G} )  \triangleq \frac{1}{\max(n,m)} {d}_{\boldsymbol{c}}^{(p)}( \mathcal{F},\mathcal{G} ) 
% \end{equation}
% \end{definition}
% \begin{definition}[Time-averaged Star-ID per-target]
% \begin{equation}
%     \bar{d}_{\boldsymbol{c},3}^{(p)}( \mathcal{F},\mathcal{G} )  \triangleq \frac{1}{(k-k_0)  \max(n,m)} {d}_{\boldsymbol{c}}^{(p)}( \mathcal{F},\mathcal{G} ) 
% \end{equation}
% \end{definition}
where $k'$ and $k$ denote %the starting and current time of the tracking, respectively, or 
the starting and ending time of a specified time-window for evaluation, respectively. %, and $N_k$ denote the number of target . 
%analogous to the definition \eqref{eq:ospa2-trackDis}, $w(t)>0$ is a positive weighting function defined for $t\in \left\{ 1,\cdots , K \right\} $ which includes all time indices from the beginning to the end of the tracking scenario, such that $\sum_{t=1}^{K}{w\left( t \right)}=1$ .
\end{definition}

\begin{remark}
In comparing with the normalizatied metrics such as OSPA and OSPA$^{(2)}$, the proposed Star-ID like the GOSPA \cite{rahmathullah2017generalized, Angel2020GOSPA} is unnormalized and will increase basically with the increase of the size of the T-FoT sets. This is reasonable since the more targets/estimates are involved, the higher cumulative error the tracker is supposed to have. The time-averaged Star-ID (TA-Star-ID) calculates the Star-ID per time unit, which evaluates the increase rate of the Star-ID over time, demonstrating the average instantaneous performance of the tracker. That being said, one may consider dividing the error by the number of targets (or that of the T-FoT estimates) for normalization. This, however, will not work since different FA and MD trajectories/segments may correspond to different durations and the FA/MD penalties depend on their duration. % (multiplied by the relevant cutoff parameter). 
As such, FA and MD are no more a simple cardinality mismatch issue. This is a significant difference of the spatio-temporial metric with the spatial metric. 
\end{remark}

\subsection{Interpreting Parameters}

\subsubsection{cut-off parameters}
%The role of the penalty factor for FAs and MDs in Star-ID is 
%Analogous to the cut-off parameter $c$ used in the OSPA, GOSPA and $\text{OSPA}^{(2)}$, the 
The segment cutoff coefficients \(c_{\text{SFA}}\) and \(c_{\text{SMD}}\) emphasize the discrepancies per time unit between matched trajectory pairs, whereas the trajectory cutoff coefficients $c_{\text{TFA}}$ and $c_{\text{TMD}}$ are designed to penalize the errors per time unit arising from unmatched trajectories due to incorrect or absent estimates. %To interpret these parameters, we reconsider 
As addressed, $c_{\text{SFA}}=c_{\text{SMD}}>0$ and $c_{\text{TFA}}=c_{\text{TMD}}>0$ are necessary for the Star-ID to be a distance. In this case, we refer to $c_{\text{SFA}} =c_{\text{SMD}} =c_\text{S}$ by default as the segment cutoff coefficient and $c_{\text{TFA}}= c_{\text{TMD}} =c_\text{T}$ by default as the trajectory cutoff coefficient. They are defined in the Euclid state space. %, regarding one time unit. 
%One further note is in order. The final penalty for both segment and trajectory FA or MD is given by the multiplication of the relevant cutoff coefficient with the duration of the FA or MD. This is different from the penalty to FA and MD in the OSPA, GOSPA and $\text{OSPA}^{(2)}$. %are defined as a the penalty for a FA or MD per time unit. 
Let us reconsider actual T-FoT set $\mathcal{F} =\{{{f}_{1}},{{f}_{2}},\ldots ,{{f}_{ m}}\}$, estimated T-FoT set $ \mathcal{G} =\{{{g}_{1}},{{g}_{2}},\ldots ,{{g}_{n}}\}$ and the global T2T association $\theta \in \Theta$. % $\Theta $ to represent the set of all reasonable/acceptable T2T associations, . %For two matched trajectories $f(t), g(t)$, $\mathbf{1}_{\theta}\left(f, g \right) =1 $. Otherwise, $\mathbf{1}_{\theta}\left(f, g \right) =0$.

Moreover, typically, $c_{\text{TFA}} \geq c_{\text{SFA}}, c_{\text{TMD}}\geq c_{\text{SMD}}$ to ensure that missing a whole trajectory is at least more serious than missing a part of it and wrongly estimating an FA trajectory is at least more serious than obtaining an FA segment. 
%  Hence, it is crucial to establish the FA and MD parameters with values that are smaller than the sensor track alignment error.

\subsubsection{Distance Order}
The role of the distance order $p$ in the Star-ID is similar to that in OSPA, GOSPA and  $\text{OSPA}^{(2)}$. The value of $p$ determines the sensitivity of ${d}_{\boldsymbol{c}}^{(p)}(\cdot,\cdot)$ to outlier estimates. We emphasize here two important choices. By using $p=1$, the Star-ID metric measures a first-order error in which the sum of the integral distance and FA and MD error equals the total metric. In contrast, %the more common choice
$p = 2$ will panel the outliers more significantly, leading to a smoother distance curve.

\section{Simulation Study} \label{sec:simula}
%In our simulation study, the individual target state is given by ${\mathbf{x}_{k}}={{\left[ {{p}_{x,k}},{{{\dot{p}}}_{x,k}},{{p}_{y,k}},{{{\dot{p}}}_{y,k}} \right]}^{T}}$, which is composed of the planar position $\left[{p}_{x,k},{p}_{y,k}\right]^{T}$ and velocity $\left[\dot{p}_{x,k},\dot{p}_{y,k}\right]^{T}$. 
In the simulation study, %and velocity $\left[\dot{p}_{x,k},\dot{p}_{y,k}\right]^{T}$. 
%The target trajectory is given in Fig. 
%The T-FoT parameters $\hat{\mathbf{C}}_{_{[k',k]}}^{(2)}: = \left\{ {c_0^{(1)},c_1^{(1)},c_2^{(1)},c_0^{(2)},c_1^{(2)},c_2^{(2)}} \right\}$ can be updated with the sliding time-window moves and then we can evaluate the estimated trajectories by calculating the difference between $\hat{\mathbf{C}}_{_{[k',k]}}^{(2)}$ and the parameters of the true trajectory curve, as described in the metric above. 
the proposed Star-ID and TA-Star-ID metrics will be demonstrated basing on the T-FoT approach to tracking in both single and multiple target scenarios. In the former, the ground truth trajectory is generated by using a traditional state space model while in the latter, they are generated directly by curve functions in the Cartesian coordinate. % and only the position trajectory when the metric is computed. 
For the purpose of comparison, we need to find the T-FoT of the ground truth for calculating the Star-ID and TA-Star-ID in the former scenario while in the latter we need to discretize the continuous curve trajectories into discrete points at each the measuring time instant for calculating the OSPA and OSPA$^{(2)}$. Perfect point-to-trajectory association is assumed in OSPA$^{(2)}$. 
%First, a straightforward scenario is used to explain the utilization and interpretation of the Star-ID metric and to compare it with the IMTA metric. Subsequently, a more comprehensive scenario is employed to illustrate how various parameters within the Star-ID metric can influence the resulting performance outcomes. 
For both Star-ID and TA-Star-ID, the segment/trajectory cutoff coefficients for FA and MD are set equal, namely $c_{\text{SFA}} =c_{\text{SMD}} =c_\text{S}$, $c_{\text{TFA}}= c_{\text{TMD}} =c_\text{T}$. We will test different cutoff parameters in two scenarios. 
The metric order for OSPA, OSPA$^{(2)}$ and Star-ID is chosen the same as $p=2$. Both the Sart-ID and TA-Star-ID at time $k$ are calculated over the sliding time-window $K=[k',k]$ with $k'=\textrm{max}(1,k-10)$. In other words, the trajectories are compared only in the current time-window part. % of the T-FoT fitting. % which corresponds to smooth distance curves. %In the subsequent sections of this paper, we will not analyze the influence of alternative $p$ values, but concentrate on $p = 2$ for convenience. 
We add that the Star-ID and TA-Star-ID metrics have also been used in the companion paper \cite{Li25TFoT-part2} for multi-target T-FoT estimator evaluation.  

% The target state model is built as follows:
% \begin{equation}
% x_t^{(i)} = \delta (t){ \boldmath{a_i}},
% \end{equation}
% where $x_t^{(i)}$ represents the position state of the $i$-th target at time $t$, $\delta (t)$ is a preselected time function matrix, and $ \boldmath{a_i}$ is the basis function weights, i.e., track parameters \cite{ji2019concave}. In practice, targets in the scene do not move randomly \cite{xu2017constrained}, and instead, they usually follow specific motion patterns. Therefore, the basis functions in $\delta (t)$ can represent the potential motion model of targets. For example, a quadratic polynomial
% \begin{equation}
% \delta (t) = \left[ {\begin{array}{*{20}{c}}
% 1&t&{{t^2}}&0&0&0\\
% 0&0&0&1&t&{{t^2}}
% \end{array}} \right].
% \end{equation}

% This applies to trajectories with constant velocity or constant acceleration models in Cartesian coordinates. 
\subsection{Single Maneuvering Target}
The maneuvering target tracking scenario was provided in Section 4.2.2 of \cite{Hartikainen13} and as shown in Fig.~\ref{fig:filtering}. 
%The deterministic target trajectory is the same as given in Section 4.2.2 of \cite{Hartikainen13} and as shown in Fig.~\ref{fig:filtering}.
Two Markov models were assumed for modeling the target movement with sampling step size $0.1$s. The first is given by a single linear Wiener process velocity %(WPV) 
model \cite{Sarkka13book} with insignificant process noise (zero-mean and power spectral density $0.01$). %, based on which the EKF, UKF and their corresponding RTS smoothers (namely EKS and UKS respectively) were realized. 
The other is given by a combination of it with a nonlinear CT model (using no position and velocity noise but zero-mean Gaussian turn rate noise with covariance $0.15$). %In the latter, 
%In this case, the IMM-EKF/EKS and IMM-UKF/UKS were employed for filtering/smoothing, separately. The following model transition probability matrix was used in the IMM approaches 
% % \begin{equation}
% % Tr_\text{IMM}=
% \[
% \Phi = 
% \left[ \begin{array}{cc}
% 0.9 & 0.1 \\
% 0.1 & 0.9 \\
% \end{array} \right] ,
% \]
% % \end{equation}
% with the prior model probabilities given by $[0.9,0.1]^\text{T}$.

% the target state is defined in the Cartesian position coordinate ${\mathbf{x}_{k}}=\left[{p}_{x,k},{p}_{y,k}\right]^{T}$.
The bearing observation is made on the position $\left[{p}_{x,k},{p}_{y,k}\right]^{T}$ of the target, which was given by four sensors located 
around the four corners of the scenario. 
%at $[s_{x,1},s_{y,1}]^\text{T}=[-0.5,3.5]^\text{T}$, $[s_{x,2},s_{y,2}]^\text{T}=[-0.5,-3.5]^\text{T}$, $[s_{x,3},s_{y,3}]^\text{T}=[7,-3.5]^\text{T}$ and $[s_{x,4},s_{y,4} ]^\text{T}=[7,3.5]^\text{T}$, respectively. 
The sensors are synchronous, having the same sampling step size $0.1$s. The simulation is carried out for $100$ %Monte Carlo 
runs, each lasting $100$s. The observation function of sensor $i=1,2,3,4$ is 
\begin{equation} \label{eq:bearing}
y_{k,i}=\text{arctan}\bigg(\frac{p_{y,k}-s_{y,i}}{p_{x,k}-s_{x,i}}\bigg) + v_{k,i} ,
\end{equation}
where $v_{k,i} \sim \mathcal{N}(0,0.0036  \text{rad}^2)$. %\Sigma_v
The T-FoT approach was carried out in $x$ and $y$ dimensions individually, with a sliding time-window $K$ no longer than one second. 
The polynomial T-FoT of order $\gamma=1$ is given as 
\begin{equation} 
\begin{cases}
p_{x,t}=c^{(1)}_0+c^{(1)}_1t,\\
p_{y,t}=c^{(2)}_0+c^{(2)}_1t. 
\end{cases}
\end{equation}
Given that the four sensors are of the same quality which is time-invariant (here we do not really need to know the statistics of $v_{k,i}$), the joint optimization parameters at time $k$ are estimated by 
\begin{equation} %\label{eq:4bearingFitting}
\underset{c^{(1)}_0,c^{(1)}_1,c^{(2)}_0,c^{(2)}_1}{\text{argmin}} \sum_{t=k'}^{k} \sum_{i=1}^4 \bigg(y_{t,i}-\text{arctan}\bigg(\frac{c^{(2)}_0+c^{(2)}_1t-s_{y,i}}{c^{(1)}_0+c^{(1)}_1t-s_{x,i}}\bigg)\bigg)^2 . \nonumber 
\end{equation}
%where $k_2$ is the latest time. %The above nonlinear formula is optimized by the LS curve fitting function: LSQCURVEFIT provided with the Optimization Toolbox of the Matlab software.

% As addressed in Section \ref{sec:stf}, three forms of fitting-inference can be jointly implemented based on the same fitting function: delayed fitting (by setting $t=k-5$ which estimated the state with $0.5$s delay), online fitting (by setting $t=k$ which estimated the state using sensor data in the latest 10 time instants), and smoothed fitting (cf.\eqref{eq:Ck_fixedsmoothing}). 

\begin{figure}[htbp]
\centerline{\includegraphics[width=0.75\columnwidth]{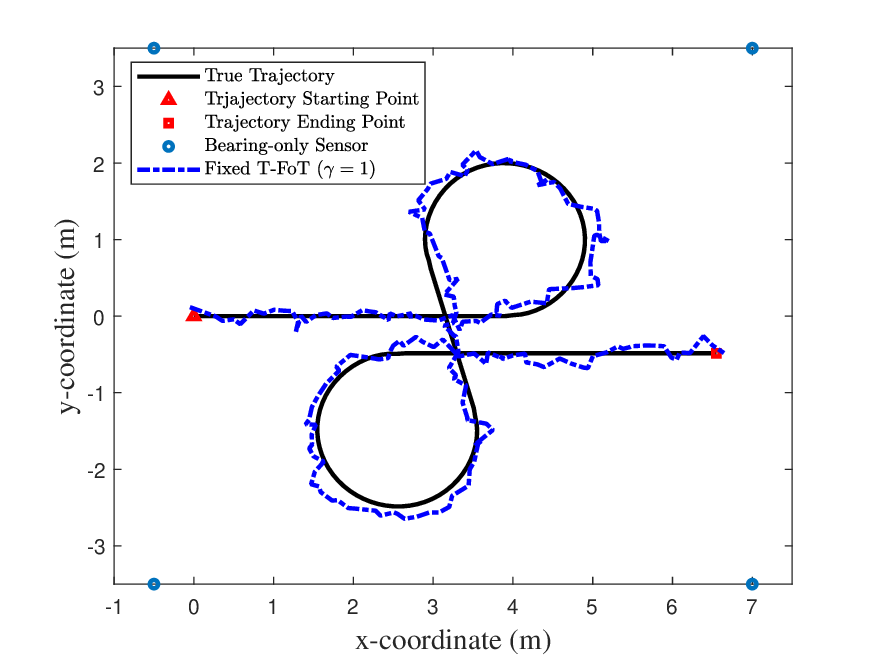}}
\caption{The real trajectory of a maneuvering target and estimates in one rial. % in the Cartesian coordinate.
}
\label{fig:filtering}
\end{figure}

\begin{figure}[htbp]
\centerline{\includegraphics[width=0.75\linewidth]{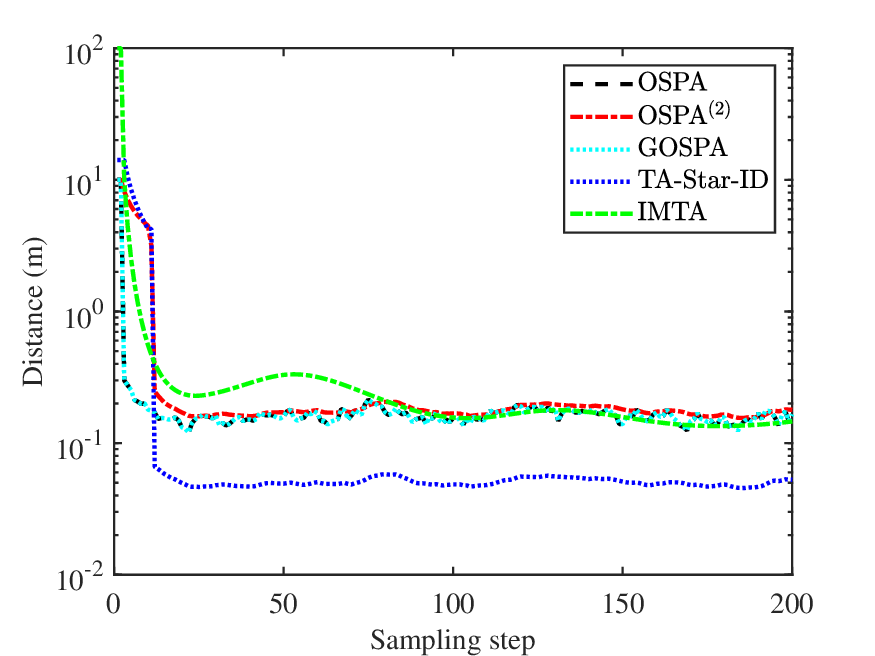}}
\caption{OSPA, $\text{OSPA}^{(2)}$, TA-Star-ID and IMTA distance of the same T-FoT.}
\label{fig:ospa_ospa2_star_id}
\end{figure}

\begin{table}[htbp]
  \centering
  \caption{Metric Parameters used in the single target scenario\label{tab:para_set}}
    \begin{tabular}{ c | c |c}
    \hline\noalign{\smallskip}
    Parameters & Value Used &  Metrics \bigstrut\\ \noalign{\smallskip}\hline
    \noalign{\smallskip}
    $c$          & 10m        & \multirow{2}[4]{*}{OSPA, $\text{OSPA}^{(2)}$, GOSPA} \bigstrut\\
     \noalign{\smallskip}
\cline{1-2} \noalign{\smallskip}   $p$          & 2          &  \bigstrut\\
\noalign{\smallskip}\hline\noalign{\smallskip}
    $\alpha $     & 2          & GOSPA \bigstrut\\
     \noalign{\smallskip}\hline
     \noalign{\smallskip}
  $c_\text{S}$        & 10m        & \multirow{2}[4]{*}{IMTA, (TA-)Star-ID} \bigstrut\\
   \noalign{\smallskip}
\cline{1-2} \noalign{\smallskip}  $c_\text{T}$         & 10m        &  \bigstrut\\
    \noalign{\smallskip}\hline 
    \end{tabular}%
  \label{tab:addlabel}%
\end{table}%

The T-FoT needs at least $\gamma +1$ data points to start the fitting. Therefore, it cannot produce any estimate in the first sampling instant when $\gamma=1$. Our fitting was performed from the second sampling instant and the first instant is treated as an MD. 
The simulation was performed for 100 runs, 200 sampling steps per run, based on the same trajectory (as shown in Fig.~\ref{fig:filtering}) but different observation realizations. 
%Two groups of simulation were performed with 100 runs, 200 sampling steps per run, based on the same trajectory (as shown in Fig.~\ref{fig:filtering}) but different observations due to different realizations of the observation noise in each run. 

Based on the parameters given in Table \ref{tab:para_set}, the average OSPA, OSPA$^{(2)}$, GOSPA, IMTA and TA-Star-ID over time are given in Fig. \ref{fig:ospa_ospa2_star_id}.  First, the OSPA and GOSPA metrics demonstrate complete consistency, as the T-FoT accurately estimates the number of targets, i.e.,  $n-m=0$ in both Eq. \eqref{eq:ospa} and \eqref{eq:Gospa}.
The fluctuation in the TA-Star-ID over time complies with $\text{OSPA}^{(2)}$, but is smaller and smoother than the OSPA over time; the IMTA is the smoothest over time. This is simply because both the TA-Star-ID and $\text{OSPA}^{(2)}$, as well as the IMTA, account for the trajectory-segment in the time-window rather than only the estimates at the newest time. 

\subsection{Multiple Target In Presence of FA and MD}

In this scenario, four targets moved with constant velocity in the area \([-3,9] \text{km} \times [-9,15]\)km, as shown in Fig. \ref{fig:multi_scene}. The starting and stopping positions for each trajectory are labeled with triangles and boxes, respectively. % to demonstrate the impact of different parameters on the Star-ID metric. 
The existing intervals for their trajectories are \( [1,100]\)s, \( [10,75]\)s, \( [1,90]\)s and \( [5,85]\)s, respectively. 
%The observation area utilized is , and the total observation interval is \(100\)s. 
As shown, the estimated T-FoTs are given by three polynomials. Target 4 that is assumed with very low detection probability is totally missed in all time. % and target 1-3 is missed sometimes in segment. % in the form of the smooth polynomial fitting function, which are generated by T-FoTs. 

The simulation is carried out for $100$ Monte Carlo runs, each lasting $100$s. Each target trajectory is modeled by a polynomial T-FoT with fixed order $\gamma=2$ as follows 
\begin{equation}
F(t;\mathbf{C}) 
% \triangleq \left[ {\begin{array}{*{20}{c}}
% {F^{x}(t;\mathbf{C})}\\
% {F^{y}(t;\mathbf{C})}
% \end{array}} \right]
=\left[ {\begin{array}{*{20}{c}}
{c_0^{(1)},c_1^{(1)},c_2^{(1)}}\\
{c_0^{(2)},c_1^{(2)},c_2^{(2)}}
\end{array}} \right]\left[ {\begin{array}{*{20}{c}}
1\\
t\\
{{t^2}}
\end{array}} \right] . % + \mathbf{v}_k
% \left[ {\begin{array}{*{20}{c}}
% {e_t^{(1)}}\\
% {e_t^{(2)}}
% \end{array}} \right],
\end{equation}

%According to $\max\left(\frac{1}{(t_2-t_1)}d_{t_1,t_2}^{(p)}\left( f,g\right)\right)=16m$, $\beta_{t_1,t_2} \left( f,g \right)<0.8$, then the parameters are set as \(c_{\text{SFA}}=c_{\text{SMD}}=6m>\frac{16\times \sqrt{0.8}}{2\times \sqrt{2}}=5.05m \). 

The measurements are made on the position T-FoT %(modeled by a polynomial of order $\gamma=2$ in each dimension) 
with measuring step size $1$s and with additive, zero-mean Gaussian measurement noise as follows % with standard derivation $1$m, %arget originated measurements are 2-D vectors ${\mathbf{y}_{k}}={{\left[{m}_{x,k},{m}_{y,k} \right]}^{T}}$ of $x$ and $y$ position. 
%The T-FoT in each position dimension is defined as follows: 
\begin{equation}
% \left[ {\begin{array}{*{20}{c}}
% {y_t^{(1)}}\\
% {y_t^{(2)}}
% \end{array}} \right]
\mathbf{y}_k = 
F(k;\mathbf{C})
% \left[ {\begin{array}{*{20}{c}}
% {c_0^{(1)},c_1^{(1)},c_2^{(1)}}\\
% {c_0^{(2)},c_1^{(2)},c_2^{(2)}}
% \end{array}} \right]\left[ {\begin{array}{*{20}{c}}
% 1\\
% t\\
% {{t^2}}
% \end{array}} \right] 
+ \mathbf{v}_k
% \left[ {\begin{array}{*{20}{c}}
% {e_t^{(1)}}\\
% {e_t^{(2)}}
% \end{array}} \right],
\end{equation}
where $\mathbf{v}_k \sim   \mathcal{N}(\mathbf{0}, 10000\text{m}^2)$. 
%$e_{t}^{(r)}$ is the fitting error at time $t \in K$ in dimension $r=1,2$. 

\begin{figure}[htbp]
\centerline{\includegraphics[width=0.75\columnwidth]{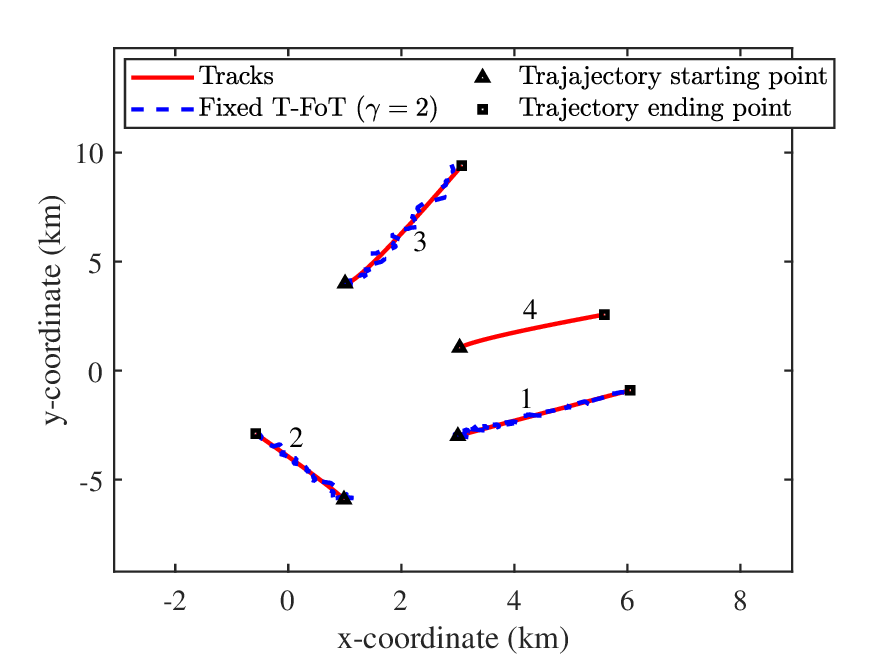}}
\caption{4 Targets move with constant velocity in the Cartesian coordinate.}
\label{fig:multi_scene}
\end{figure}

We first set fixed $c_\text{T}=1$km and different $c_\text{S}$ ($500$m, $1$km, $1.5$km and $2$km) to calculate the TA-Star-ID and Star-ID. 
As illustrated in Fig. \ref{fig:SFASMD} (a), the Star-ID increases with the increase of $c_\text{S}$ over all when there is any segment MD or FA in the concerning time-window. %The fluctuation in the Star-ID over time are consistent for different $c_\text{S}$.  
Furthermore, %it is clear that $c_\text{S}$ has the most significant impact at the initial moment when segment MD (FA) occurs. 
both the Star-ID and the TA-Star-ID increase when new targets appear and decrease with their disappearance as they are basically cumulative distances summing up the errors regarding all target estimates. % after 75s and 90s. 
%More specially, as shown in Fig. \ref{fig:SFASMD} (b), with the growth of the sliding time-window grows when $t\le 20$s, the TA-Star-ID increases with the increase of $c_\text{S}$. This increase stops when the length of the time-window is fixed, namely $t> 20$s. This behavior can be attributed to the calculation of the average distance within a sliding window. %Specifically, in the subsequent sliding window, the absence of SFA and SMD leads to constant values.
% When targets 4 and 3 disappear at $85$s and $90$s respectively, the Star-ID drops sharply at these points. Furthermore, there is a decrease in the distance at $96$s, as a segment FA to target 1 concludes at $96$s. These clearly demonstrate that the Star-ID is sensitive to the target appears and disappears. 

% Within the timeframe of the $86-90s$ , a discernible downward trajectory is observable in the curve. This phenomenon arises from Star-ID's computation of the target's distance from its birth up to the present instant. After the $86s$, there is still an error due to the FA of target 4.

We next fix $c_\text{S}=1$km and set different $c_\text{T}$ from $500$m to $2$km. As shown in Fig. \ref{fig:TMDTFA}, the TA-Star-ID and Star-ID consistently increase with the increase of $c_\text{T}$ as long as there is TFA or TMD. %After $85$s when target 4 (as well as the TMD) vanishes at $86$s, $c_\text{T}$ no longer affects the TA-Star-ID and Star-ID. 

To conclude, it is confirmed that whatever $c_\text{S}$ and $c_\text{T}$ are, the Star-ID, but not the TA-Star-ID, basically increase with the growth of the sliding time-window at the initial stage since it is a distance integrating the errors in the whole time-window. In contrast, the TA-Star-ID shows the average T-FoT estimation error for each sampling period in the time-window, showing an average performance.   

\begin{figure*}[htbp]
\centerline{\includegraphics[width=15cm,height=5.5cm]{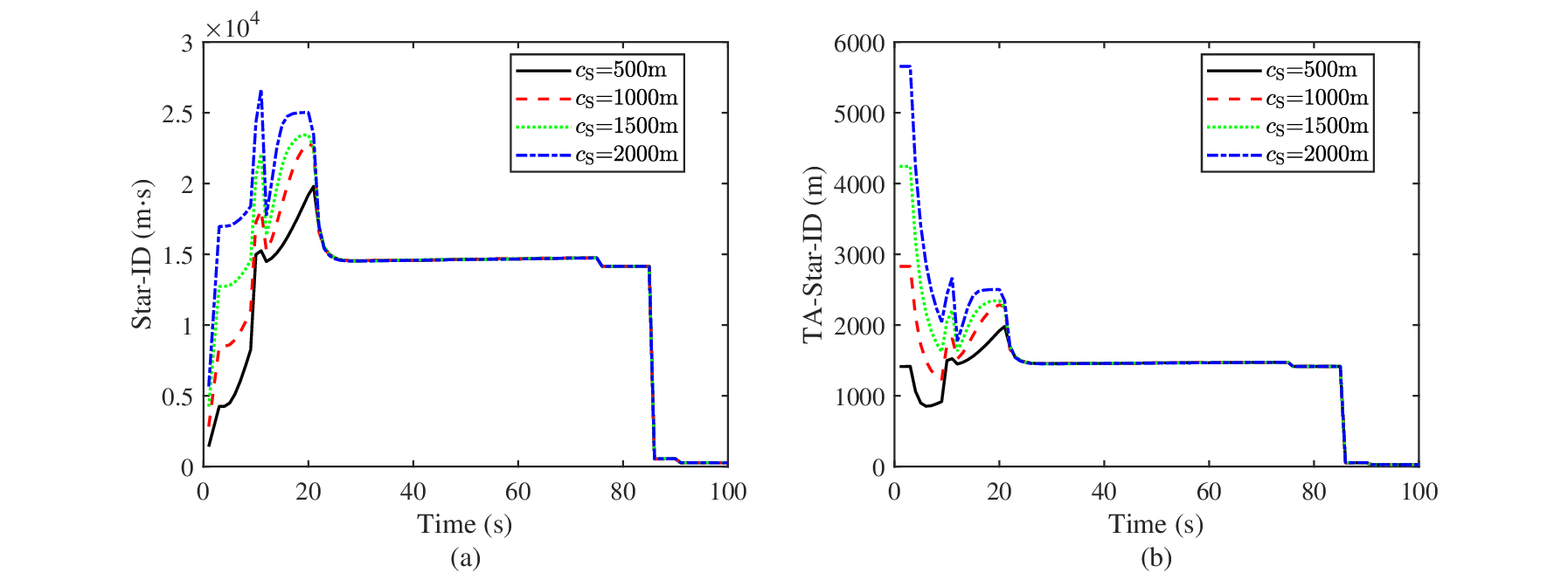}}
\caption{The effect of $c_\text{S}$ on the Star-ID (a) and TA-Star-ID (b), respectively ($c_\text{T}=1$km).}
\label{fig:SFASMD}
\end{figure*}

\begin{figure*}[htbp]
\centerline{\includegraphics[width=15cm,height=5.5cm]{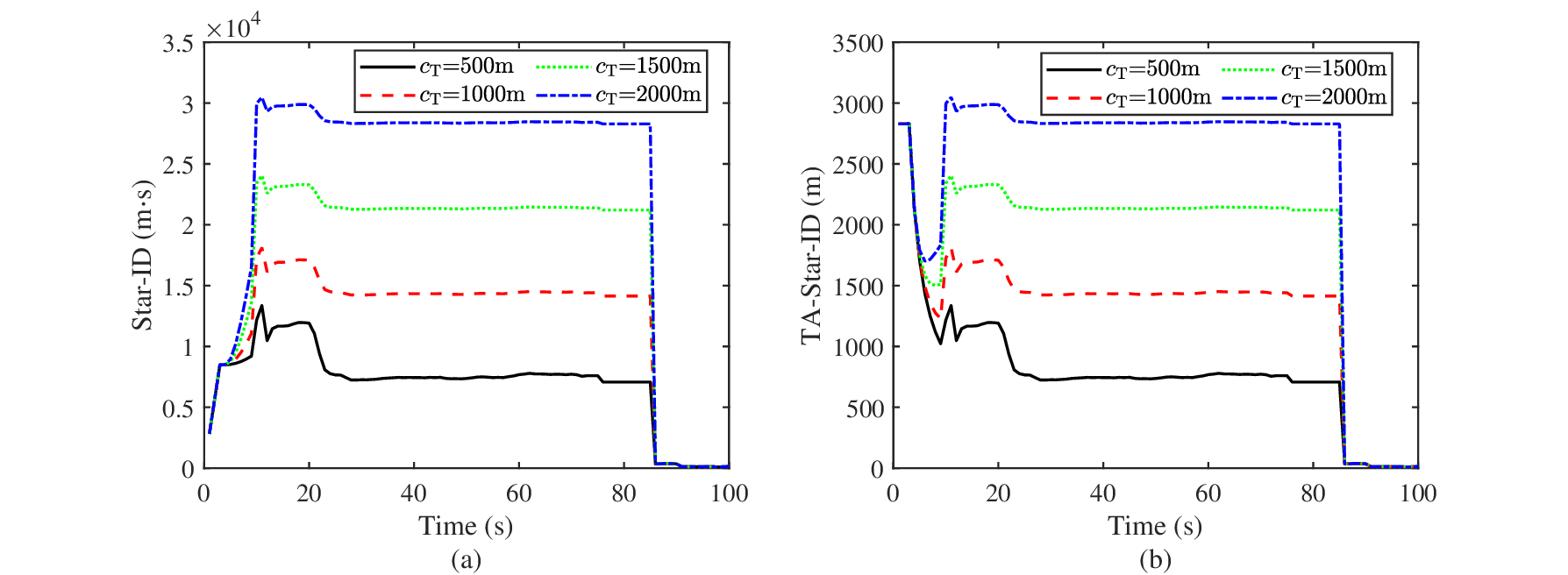}}
\caption{The effect of $c_\text{T}$ on the Star-ID (a) and TA-Star-ID (b), respectively ($c_\text{S}=1$km).}
\label{fig:TMDTFA}
\end{figure*}

%Certainly, it is essential to ensure that there are no unjustifiable disparities during the parameter setting process. Parameter setting issues do not impact the disparity between estimates and truth.

\section{Conclusion} \label{sec:con}
In this paper, we propose a metric for calculating the similarity between two arbitrary sets of continuous-time trajectories defined in the spatio-temporal domain, termed as the Star-ID. 
% which takes into account the length similarity between the between-set  trajectories in addition to their spatio-temporal-aligned localization distance and spatio-temporal-unaligned trajectory/segment mismatch.  
%error, false alarm and miss-detection penalties. 
It associates and aligns the estimated and real trajectories in the spatio-temporal domain and distinguishes between the aligned and unaligned segments in calculating either localization integral distance or FA/MD penalties. %For the aligned segments, the arc length of the trajectory is taken into account while for the unaligned segments, 
Particularly, the time length is taken into account in the FA and MD penalties.
%To calculate the Star-ID, the trajectories are pairwise associated between the two sets, aligned temporally between the associated trajectories that are decomposed into aligned and unaligned segments according to their temporal-alignment for similarity calculation. 
The trajectory association depends on the Star-ID between two trajectories, which as well as the FA and MD cutoff coefficients is the prerequisite for calculating the Star-ID between two sets of multiple trajectories. 
Conditions that ensue the metric be a distance are specified and proven. Simulation in both single and multiple target tracking scenarios has demonstrated the effectiveness of the Star-ID and its time-averaged version. %The metric, however, is simple to compute and flexible to capture many important aspects of tracking performance between sets of continuous-time trajectories. Yet, no metric is perfect and there an interesting
A valuable extension of the metric is to take into account the uncertainty or the existing probability of the estimated trajectories which may lead to better performance in both the tracklet association and distance calculation. %, similar like what has been done in \cite{Nagappa2011ospaUncertainty,he2013track} over the OSPA. 

%The effectiveness of the proposed distance metric is discussed and verified through theoretical analysis and numerical examples of a single target or multiple targets.

%By considering both the distance and shape errors between the estimated and real trajectories, the Star-ID metric offers a comprehensive evaluation of tracking algorithm performance, covering a wide range of evaluation criteria. It is a shape-based metric that can be used to evaluate the performance of tracking algorithms in various scenarios. It is particularly useful for managing differences in trajectory shape and scale. Overall, the Star-ID metric serves as a valuable tool for researchers and practitioners in the field of target tracking.

%Furthermore, it is important to highlight that the Star-ID metric has limitations in its applicability. It is primarily designed for polynomial T-FoT implementations and may not be directly applicable to other ways of continuous-time functions, such as Gaussian process and neural network. This presents another challenge that we will address in future work. 

\section*{Appendix}
\subsection[A]{Minkowski's inequality for T-FoT Sets} 
%for  $1\le p< \infty $. We use this result several times in our proof.
\begin{proposition}\label{Proposition-Apped-A}
Let $d^{(\alpha)}(f, g)$ and $d^{(\beta)}(f, g)$ denote two different metrics for measuring the distance between trajectories $f(t)$ and $g(t)$. Given that both metrics satisfy the triangle inequality, i.e., $d^{(\alpha)}(f, g)\le d^{(\alpha)}(f, h)+d^{(\alpha)}(h, g), d^{(\beta)}(f, g)\le d^{(\beta)}(f, h)+d^{(\beta)}(h, g)$ for any trajectory $h(t)$, the new metric $d^{(\alpha,\beta)}(f,g)$ defined as, $1\le p< \infty,$ 
$$
d^{(\alpha,\beta)}(f,g) \triangleq \Big(d^{(\alpha)}(f, g)^p+d^{(\beta)}(f, g)^p\Big)^{\frac{1}{p}}
$$
%Therefore, $d^{(\alpha,\beta)}(f,g)$ also 
satisfies the triangle inequality:
\begin{equation}
  d^{(\alpha,\beta)}(f,g)\le d^{(\alpha,\beta)}(f,h)+d^{(\alpha,\beta)}(h,g)
\end{equation} \label{eq23}
\end{proposition}

\begin{proof}
The proof is straightforward as follows
\begin{align}
   &d^{(\alpha,\beta)}(f,h)+d^{(\alpha,\beta)}(h,g) \nonumber \\
=&\big(d^{(\alpha)}(f,h)^p+d^{(\beta)}(f,h)^p\big)^{\frac{1}{p}}+\big(d^{(\alpha)}(h,g)^p+d^{(\beta)}(h,g)^p\big)^{\frac{1}{p}} \nonumber \\
\ge & \Big( \big(d^{(\alpha)}(f,h)+ d^{(\alpha)}(h,g)\big)^p+\big(d^{(\beta)}(f,h)+ d^{(\beta)}(h,g)\big)^p \Big)^{\frac{1}{p}} \label{po_3} \\
\ge & \Big( d^{(\alpha)}(f,g)^p+d^{(\beta)}(f,g)^p \Big)^{\frac{1}{p}} \nonumber \\
=& d^{(\alpha,\beta)}(f,g)
\end{align}
where the Minkowski's inequality \footnote{Given two sequences $x=(\xi_1,\dots ,\xi_n)$ and $y=(\nu_1,\dots ,\nu_n)$, the Minkowski's inequality is defined as follows, for  $1\le p< \infty $,
\begin{equation}
\left( \sum_{i=1}^n{\left| \xi _i+\nu_i \right|^p} \right)^{\frac{1}{p}}  \leq \left(\sum_{i=1}^n{\left| \xi _i \right|^p} \right)^{\frac{1}{p}} + \left( \sum_{i=1}^n{\left| \nu_i \right|^p} \right)^{\frac{1}{p}} \nonumber
\end{equation}
} \cite{kubrusly2011elements} was used in \eqref{po_3}. %for the three relation, and the $d^{(\alpha)}(f, g), d^{(\beta)}(f, g) $ triangle inequality for fourth relation.
\end{proof}

\subsection[B]{Proof of Theorem \ref{theorem_d(f,g)}} \label{sec:proof-d(f,g)}

Firstly, since \( d_{t_1,t_2}^{(p)}\left( f,g\right) \ge 0\), when $c_{\text{SFA}},c_{\text{SMD}} >0 $,  $d_{\boldsymbol{c}}^{(p)}\left( f,g \right)$ satisfies non-negativity. Second, for any \(f \), \(g\), it is easy to see that \(d_{\boldsymbol{c}}^{(p)}\left( f,g \right)=0\) if and only if \(f=g \). Thus, the identity property is satisfied. Third, when $c_{\text{SFA}}=c_{\text{SMD}}$ and $c_{\text{TFA}}=c_{\text{TMD}}$ ,  \(d_{\boldsymbol{c}}^{(p)}( \cdot , \cdot )\) is symmetric in its arguments, satisfying the third property. 

Lastly, we prove the triangle inequality by starting from the following definition $\overline{d}_{t_1,t_2}^{(\boldsymbol{c},p)} (f,g) \triangleq \min \Big(d_{t_1,t_2}^{(p)} (f,g) , r^{\frac{1}{p}}(c_{\text{SFA}}+c_{\text{SMD}})(t_2-t_1)  \Big)$. It can be straightforwardly asserted that
\begin{equation}
 d_{t_1,t_2}^{(\boldsymbol{c},p)} (f,g)  = \left(\overline{d}_{t_1,t_2}^{(\boldsymbol{c},p)} (f,g) \right)^p . 
\end{equation}

% To facilitate the proof, we reformulate the expression of $d_{t_1,t_2}^{(\boldsymbol{c},p)} (f,g)$ as follows:
% \begin{align}\
% d_{t_1,t_2}^{(\boldsymbol{c},p)} (f,g) 
% &=\min \Big(\beta_{t_1,t_2} \left( f,g \right) \big(d_{t_1,t_2}^{(p)} (f,g)\big)^p , \nonumber \\ 
% &~~~r(c_{\text{SFA}}+c_{\text{SMD}})^p(t_2-t_1)^p  \Big)\nonumber \\ 
% &=\bigg(\min \Big(\big(\beta_{t_1,t_2} \left( f,g \right) \big)^{\frac{1}{p}} d_{t_1,t_2}^{(p)} (f,g) , \nonumber \\
% &~~~r^{\frac{1}{p}}(c_{\text{SFA}}+c_{\text{SMD}})(t_2-t_1)  \Big)\bigg)^p \nonumber \\ 
% &=\left(\overline{d}_{t_1,t_2}^{(\boldsymbol{c},p)} (f,g) \right)^p
% \end{align}
% here, $\overline{d}_{t_1,t_2}^{(\boldsymbol{c},p)} (f,g) \triangleq =\min \Big(\big(\beta_{t_1,t_2} \left( f,g \right) \big)^{\frac{1}{p}} d_{t_1,t_2}^{(p)} (f,g) , r^{\frac{1}{p}}(c_{\text{SFA}}+c_{\text{SMD}})(t_2-t_1)  \Big)$.

%First we proof $d_{\boldsymbol{c}}^{(p)}\left( f,g \right)$  is a metric. The first three properties have been proven above, we only prove the triangle inequality. 

Considering three T-FoTs $f$, $g$ and $h$ with time intervals denoted by
$T^{(f)}=[ t_s^{(f)}, t_e^{(f)} ]$, $T^{(g)}=[ t_s^{(g)}, t_e^{(g)} ]$ and $T^{(h)}=[ t_s^{(h)}, t_e^{(h)} ]$, respectively, their joint temporal span of existence is $T^{(f,h,g)}=\left[\min\left(t_s^{(f)}, t_s^{(g)}, t_s^{(h)}\right), \max\left(t_e^{(f)},  t_e^{(g)}, t_e^{(h)}\right)\right]$. Our proof is categorized into the following complementary cases:

Case 1: $t_s^{(h)} < t_s^{(f)} <t_s^{(g)}$, $t_e^{(h)} < t_e^{(f)} <t_e^{(g)}$,
\begin{align}
&d_{\boldsymbol{c}}^{(p)}\left( f,g \right) \nonumber \\
=& \left(r\left(c_{\text{S}}(t_s^{(g)}-t_s^{(f)})+c_{\text{S}}(t_e^{(g)}-t_e^{(f)})\right)^p +d_{t_s^{(g)},t_e^{(f)}}^{(\boldsymbol{c},p)} (f,g)\right)^{\frac{1}{p}} \nonumber\\
\le & \bigg(r\left(c_{\text{S}}(t_s^{(g)}-t_s^{(f)})+c_{\text{S}}(t_e^{(g)}-t_e^{(f)})\right)^p \nonumber \\
&~~~~+\Big[ \overline{d}_{t_s^{(f)},t_e^{(h)}}^{(\boldsymbol{c},p)} (f,g) +r^{\frac{1}{p}}c_{\text{S}}(t_e^{(f)}-t_e^{(h)})\Big]^p\bigg)^{\frac{1}{p}} \label{eq:Append-B-2} \\
\le & \bigg(r\Big(c_{\text{S}}(t_s^{(g)}-t_s^{(f)})+c_{\text{S}}(t_e^{(g)}-t_e^{(f)}) \nonumber \\
&~~~~+2c_{\text{S}}(t_e^{(f)}-t_e^{(h)})\Big)^p+\Big[ \overline{d}_{t_s^{(f)},t_e^{(h)}}^{(\boldsymbol{c},p)} (f,g) \Big]^p\bigg)^{\frac{1}{p}} \nonumber  \\
\le & \bigg(r\Big[ \left(c_{\text{S}}(t_s^{(f)}-t_s^{(h)})+c_{\text{S}}(t_s^{(g)}-t_s^{(h)})\right) \nonumber  \\
&~~~~~~+\left(c_{\text{S}}(t_e^{(f)}-t_e^{(h)})\right)+\left(c_{\text{S}}(t_e^{(g)}-t_e^{(h)})\right)\Big]^p \nonumber  \\
&~+\Big[ \overline{d}_{t_s^{(f)},t_e^{(h)}}^{(\boldsymbol{c},p)} (f,h) +\overline{d}_{t_s^{(g)},t_e^{(h)}}^{(\boldsymbol{c},p)} (h,g) \Big]^p\bigg)^{\frac{1}{p}} \nonumber \\
\le& \bigg( r\left(c_{\text{S}}(t_s^{(f)}-t_s^{(h)})+c_{\text{S}}(t_e^{(f)}-t_e^{(h)})\right)^p +d_{t_s^{(f)},t_e^{(h)}}^{(\boldsymbol{c},p)} (f,h) \bigg)^{\frac{1}{p}} \nonumber  \\
+&\bigg(r\left(c_{\text{S}}(t_s^{(g)}-t_s^{(h)})+c_{\text{S}}(t_e^{(g)}-t_e^{(h)})\right)^p +d_{t_s^{(g)}, t_e^{(h)}}^{(\boldsymbol{c},p)} (h,g) \bigg)^{\frac{1}{p}} \label{eq:Append-B-5}   \\
=&d_{\boldsymbol{c}}^{(p)}\left( f,h \right)+d_{\boldsymbol{c}}^{(p)}\left( h,g \right) \nonumber
\end{align}
where \eqref{dcpt1t2} %or \eqref{eq:SFASMD-interpration} in the interval $[t_e^{(h)}, t_e^{(f)}]$
was used in \eqref{eq:Append-B-2} while 
Proposition \ref{Proposition-Apped-A} in \eqref{eq:Append-B-5}.

\begin{align}
&d_{\boldsymbol{c}}^{(p)}\left( f,h \right) \nonumber \\
=& \bigg( r\left(c_{\text{S}}(t_s^{(f)}-t_s^{(h)})+c_{\text{S}}(t_e^{(f)}-t_e^{(h)})\right)^p +d_{t_s^{(f)},t_e^{(h)}}^{(\boldsymbol{c},p)} (f,h) \bigg)^{\frac{1}{p}} \nonumber \\
\le & \bigg(r\left(c_{\text{S}}(t_s^{(f)}-t_s^{(h)})+c_{\text{S}}(t_e^{(f)}-t_e^{(h)})\right)^p  \nonumber\\
&~+\Big[ \overline{d}_{t_s^{(g)},t_e^{(h)}}^{(\boldsymbol{c},p)} (f,h) +r^{\frac{1}{p}}c_{\text{S}}(t_s^{(g)}-t_s^{(f)})\Big]^p\bigg)^{\frac{1}{p}} \label{eq:Append-B-2-1} \\
\le & \bigg(r\Big[ \left(c_{\text{S}}(t_s^{(g)}-t_s^{(f)})+c_{\text{S}}(t_s^{(g)}-t_s^{(h)})\right) \nonumber\\
&~~~~~~+\left(c_{\text{S}}(t_e^{(g)}-t_e^{(h)})\right)+\left(c_{\text{S}}(t_e^{(g)}-t_e^{(f)})\right)\Big]^p \nonumber\\
&~+\Big[ \overline{d}_{t_s^{(g)},t_e^{(f)}}^{(\boldsymbol{c},p)} (f,g) +\overline{d}_{t_s^{(g)},t_e^{(h)}}^{(\boldsymbol{c},p)} (h,g) \Big]^p\bigg)^{\frac{1}{p}} \nonumber\\
\le& \bigg(r\left(c_{\text{S}}(t_s^{(g)}-t_s^{(f)})+c_{\text{S}}(t_e^{(g)}-t_e^{(f)})\right)^p +d_{t_s^{(g)},t_e^{(f)}}^{(\boldsymbol{c},p)} (f,g)\bigg)^{\frac{1}{p}}\nonumber \\
+ & \bigg(r\left(c_{\text{S}}(t_s^{(g)}-t_s^{(h)})+c_{\text{S}}(t_e^{(g)}-t_e^{(h)})\right)^p +d_{t_s^{(g)},t_e^{(h)}}^{(\boldsymbol{c},p)} (h,g)\bigg)^{\frac{1}{p}}  \label{eq:Append-B-2-2}\\
=&d_{\boldsymbol{c}}^{(p)}\left( f,g \right)+d_{\boldsymbol{c}}^{(p)}\left( h,g \right) \nonumber
\end{align}
where \eqref{dcpt1t2} was used in \eqref{eq:Append-B-2-1} while 
Proposition \ref{Proposition-Apped-A} in \eqref{eq:Append-B-2-2}.

\begin{align}
&d_{\boldsymbol{c}}^{(p)}\left( h,g \right)  \nonumber \\
=&\bigg(r\left(c_{\text{S}}(t_s^{(g)}-t_s^{(h)})+c_{\text{S}}(t_e^{(g)}-t_e^{(h)})\right)^p +d_{t_s^{(g)},t_e^{(h)}}^{(\boldsymbol{c},p)} (h,g)\bigg)^{\frac{1}{p}} \nonumber \\
\le & \bigg( r\Big[ \left(c_{\text{S}}(t_s^{(f)}-t_s^{(h)})\right)+\left(c_{\text{S}}(t_s^{(g)}-t_s^{(f)})\right) \nonumber\\
&~~~~~~+\left(c_{\text{S}}(t_e^{(f)}-t_e^{(h)})\right)+\left(c_{\text{S}}(t_e^{(g)}-t_e^{(f)})\right)\Big]^p  \nonumber\\
&~~~+\Big[ \overline{d}_{t_s^{(f)},t_e^{(h)}}^{(\boldsymbol{c},p)} (f,h) +\overline{d}_{t_s^{(g)},t_e^{(f)}}^{(\boldsymbol{c},p)} (f,g) \Big]^p\bigg)^{\frac{1}{p}} \nonumber\\
\le&d_{\boldsymbol{c}}^{(p)}\left( f,h \right)+d_{\boldsymbol{c}}^{(p)}\left( f,g \right)  \label{eq:casef_g} 
\end{align}
where Proposition \ref{Proposition-Apped-A} was used in \eqref{eq:casef_g}.

This proves the triangle inequality for this case. The proof of the triangle inequality for the remaining cases 2-6 is analogous to that for case 1  and is omitted here.

Case 2: $t_s^{(h)} < t_s^{(g)} <t_s^{(f)}$, $t_e^{(h)} < t_e^{(f)} <t_e^{(g)}$

Case 3: $t_s^{(h)} < t_s^{(f)} <t_s^{(g)}$, $t_e^{(f)} < t_e^{(h)} <t_e^{(g)}$.

Case 4: $t_s^{(h)} < t_s^{(g)} <t_s^{(f)}$, $t_e^{(f)} < t_e^{(h)} <t_e^{(g)}$. 

Case 5: $t_s^{(h)} < t_s^{(f)} <t_s^{(g)}$, $t_e^{(f)} < t_e^{(g)} <t_e^{(h)}$. 

Case 6: $t_s^{(h)} < t_s^{(g)} <t_s^{(f)}$, $t_e^{(f)} < t_e^{(g)} <t_e^{(h)}$.

Summing up, we have now proved that the triangle inequality holds in all cases for \(d_{\boldsymbol{c}}^{(p)}\left( f,g \right)\). %is a metric.  

\subsection[C]{Proof of Theorem \ref{theorem_d(F,G)} } 
\label{sec:proofStar-ID}

%From the definition of $d_{I}^{(c)}$ it is clear that $d_{I}^{(c)}(X,Y)\ge 0$ and that $d_{I}^{(c)}$ satisfies the identity and symmetry properties. We would like to prove the triangle inequality:

%for any three RFSs $X$, $Y$ and $Z$. The proof is dealt in three cases based on the values of $|X|$, $|Y|$ and $|Z|$. Since this inequality is symmetric in $X$ and $Y$, we may assume without loss of generality that $|X| \leq |Y|$.

First, as proven in Theorem \ref{theorem_d(f,g)}, $d_{\boldsymbol{c}}^{(p)}\left( f,g \right)$ satisfies non-negativity while the rest part of ${d}_{\boldsymbol{c}}^{(p)}( \mathcal{F},\mathcal{G} )$ is non-negative. So, the non-negativity property is therefore satisfied. Second, for \(f(t) \), \(g(t) \), since \(d_{\boldsymbol{c}}^{(p)}\left( f,g \right)=0\) if and only if \(f(t)=g(t) \), then \( {d}_{\boldsymbol{c}}^{(p)}( \mathcal{F},\mathcal{G} ) =0 \) if and only if \(\mathcal{F}=\mathcal{G}\). Thus, the identity property is satisfied. Third, when $c_{\text{SFA}}=c_{\text{SMD}}$ and $c_{\text{TFA}}=c_{\text{TMD}}$,  \(d_{\boldsymbol{c}}^{(p)}( \cdot , \cdot )\) is symmetric in its arguments, \({d}_{\boldsymbol{c}}^{(p)}( \mathcal{F},\mathcal{G} )\) is also symmetric in its arguments, satisfying the third property. %The proof of the triangle inequality is provided in Appendix.
In what follows, we would like to prove the triangle inequality:
\begin{equation} \label{eq:triangle}
{d}_{\boldsymbol{c}}^{(p)}( \mathcal{F},\mathcal{G} )\le {d}_{\boldsymbol{c}}^{(p)}\left( \mathcal{F},\mathcal{H} \right)+{d}_{\boldsymbol{c}}^{(p)}\left( \mathcal{H},\mathcal{G} \right)
\end{equation}

Since this inequality is symmetric, % in \(\mathcal{F}\) and \(\mathcal{G}\), 
we may assume without loss of generality that \(m\leq n \). %\eqref{eq:Star_ID_case2}. 
Our proof (given in the next page) addresses three complete and complementary cases differing in the relationship between \(m\), \(n\) and \(l\). Proposition 1 was used in \eqref{eq:Star_ID_case1}, \eqref{eq:Star_ID_case2} and \eqref{eq:Star_ID_case3}. As shown in all these three cases, the triangle inequality \eqref{eq:triangle} holds. 

\begin{table*} 
{\normalsize
Case 1: \(m \leq n\leq l\)
 \begin{align}
{d}_{\boldsymbol{c}}^{(p)}( \mathcal{F},\mathcal{G} ) &\le \bigg( \sum_{i=1}^n{r \left( c_{\text{T}} T^{\theta}_{\text{TFA,}i} \right) ^p}   +\sum_{j=1}^m\sum_{i=1}^n{\mathbf{1}_{\theta}\left( f_j ,g_i  \right)   \Big({d}_{\boldsymbol{c}}^{(p)}\left( f_j ,g_i \right) \Big) ^p} \bigg)^{\frac{1}{p}}   \nonumber \\
 &\le \bigg(  \sum_{i=1}^n{r \left( c_{\text{T}} T^{\theta}_{\text{TFA,}i} \right) ^p} \nonumber  +\sum_{j=1}^m \sum_{i=1}^l\Big[ \mathbf{1}_{\theta} \left( f_j ,h_i  \right) 
\left({d}_{\boldsymbol{c}}^{(p)}\left( f_j ,h_i\right)\right)  +\mathbf{1}_{\theta} \left( g_j ,h_i  \right) 
\left({d}_{\boldsymbol{c}}^{(p)}\left( g_j ,h_i\right)\right)\Big]  ^p \bigg)^{\frac{1}{p}}   \nonumber\\
 &\le \bigg(  \sum_{i=1}^n{r \left( c_{\text{T}} T^{\theta}_{\text{TMD,}i} \right) ^p}+\sum_{i=1}^l{r \left( c_{\text{T}} T^{\theta}_{\text{TFA,}i} \right) ^p}  +\sum_{j=m+1}^n\sum_{i=1}^l{\mathbf{1}_{\theta}\left( g_j ,h_i  \right)   \Big({d}_{\boldsymbol{c}}^{(p)}\left( g_j ,h_i \right)\Big)  ^p}  \nonumber\\
 &~~~~~~ +\sum_{j=1}^m \sum_{i=1}^l\Big[ \mathbf{1}_{\theta} \left( f_j ,h_i  \right) 
\left({d}_{\boldsymbol{c}}^{(p)}\left( f_j ,h_i\right)\right)+\mathbf{1}_{\theta} \left( g_j ,h_i  \right) 
\left({d}_{\boldsymbol{c}}^{(p)}\left( g_j ,h_i\right)\right)\Big]  ^p \bigg)^{\frac{1}{p}}  \nonumber \\
 &\le \bigg(\sum_{i=1}^l{r \left( c_{\text{T}} T^{\theta}_{\text{TFA,}i} \right) ^p}  +\sum_{j=1}^m\sum_{i=1}^l{\mathbf{1}_{\theta}\left( f_j ,h_i  \right)   \Big({d}_{\boldsymbol{c}}^{(p)}\left( f_j ,h_i \right)\Big)  ^p} \bigg)^{\frac{1}{p}}   \nonumber\\
 &~~~~~~+ \bigg( \sum_{i=1}^n{r \left( c_{\text{T}} T^{\theta}_{\text{TMD,}i} \right) ^p} +\sum_{j=1}^n\sum_{i=1}^l{\mathbf{1}_{\theta}\left( g_j ,h_i  \right)   \Big({d}_{\boldsymbol{c}}^{(p)}\left( g_j ,h_i \right) \Big) ^p} \bigg)^{\frac{1}{p}}  \label{eq:Star_ID_case1} \\
&= {d}_{\boldsymbol{c}}^{(p)}\left( \mathcal{F},\mathcal{H} \right)+{d}_{\boldsymbol{c}}^{(p)}\left( \mathcal{G},\mathcal{H} \right)  \nonumber
\end{align}
%where Proposition 3 was used in \eqref{eq:Star_ID_case1}.
%using the Proposition \ref{Proposition-Apped-A} for the fourth relation.
}

{\normalsize
Case 2: \(m\leq  l\leq n\)
\begin{align}
{d}_{\boldsymbol{c}}^{(p)}( \mathcal{F},\mathcal{G} )  &\le\bigg( \sum_{i=1}^n{ r \left( c_{\text{T}} T^{\theta}_{\text{TFA,}i} \right) ^p}  +\sum_{j=1}^m\sum_{i=1}^n{\mathbf{1}_{\theta}\left( f_j ,g_i  \right)   \Big({d}_{\boldsymbol{c}}^{(p)}\left( f_j ,g_i \right) \Big) ^p}\bigg)^{\frac{1}{p}}  \nonumber\\
&\le \bigg( \sum_{i=1}^n{ r \left( c_{\text{T}} T^{\theta}_{\text{TFA,}i} \right) ^p}+\sum_{i=1}^l{ r \left( c_{\text{T}} T^{\theta}_{\text{TFA,}i} \right) ^p}  +\sum_{j=1}^m \sum_{i=1}^n\Big[ \mathbf{1}_{\theta} \left( f_j ,h_i  \right) 
\left({d}_{\boldsymbol{c}}^{(p)}\left( f_j ,h_i\right)\right) +\mathbf{1}_{\theta} \left( h_j ,g_i  \right) 
\left({d}_{\boldsymbol{c}}^{(p)}\left( h_j ,g_i\right)\right)\Big] ^p  \nonumber\\
&~~~~~~~~+\sum_{j=m+1}^l\sum_{i=1}^n \mathbf{1}_{\theta}\left( h_j ,g_i  \right)  \Big( {d}_{\boldsymbol{c}}^{(p)}\left( h_j ,g_i\right) \Big)^p \bigg)^{\frac{1}{p}}  \nonumber\\
  &\le \bigg( \sum_{i=1}^l{ r \left( c_{\text{T}} T^{\theta}_{\text{TFA,}i} \right) ^p}  +\sum_{j=1}^m\sum_{i=1}^l \mathbf{1}_{\theta}\left( f_j ,h_i  \right)  \Big( {d}_{\boldsymbol{c}}^{(p)}\left( f_j ,h_i\right) \Big) ^p \bigg)^{\frac{1}{p}}   \nonumber \\
 &~~~~~~+\bigg(\sum_{i=1}^n{ r \left( c_{\text{T}} T^{\theta}_{\text{TFA,}i} \right) ^p} +\sum_{j=1}^l\sum_{i=1}^n{\mathbf{1}_{\theta}\left( h_j ,g_i  \right)   \Big( {d}_{\boldsymbol{c}}^{(p)}\left( h_j ,g_i \right) \Big) ^p} \bigg)^{\frac{1}{p}} \label{eq:Star_ID_case2}\\
  &= {d}_{\boldsymbol{c}}^{(p)}\left( \mathcal{F},\mathcal{H} \right)+{d}_{\boldsymbol{c}}^{(p)}\left( \mathcal{G},\mathcal{H} \right)  \nonumber
\end{align}
%From here, we can argue similar to the case 1 which 
%where Proposition 3 was used in \eqref{eq:Star_ID_case2}. 
}

{\normalsize
Case 3: \(l\leq  m\leq n\)
\begin{align}
{d}_{\boldsymbol{c}}^{(p)}( \mathcal{F},\mathcal{G} )&\le\Bigg(\sum_{i=1}^n{r \left( c_{\text{T}} T^{\theta}_{\text{TFA,}i} \right) ^p}   +\sum_{j=1}^m\sum_{i=1}^n{\mathbf{1}_{\theta}\left( f_j ,g_i  \right)   \Big({d}_{\boldsymbol{c}}^{(p)}\left(  f_j ,g_i \right) \Big) ^p}\Bigg)^{\frac{1}{p}} \nonumber  \\
&\le \Bigg( \sum_{i=1}^n{r \left( c_{\text{T}} T^{\theta}_{\text{TFA,}i} \right) ^p} +\sum_{i=1}^m{r \left( c_{\text{T}} T^{\theta}_{\text{TFA,}i} \right) ^p} +\sum_{j=1}^l\sum_{i=1}^n{\mathbf{1}_{\theta}\left( f_j ,g_i  \right)   \Big({d}_{\boldsymbol{c}}^{(p)}\left(  f_j ,g_i\right) \Big) ^p}\Bigg)^{\frac{1}{p}} \label{eq:Case-3-prof}  \\
%\end{align*}
%To get the above inequality, for $j=l+1,\cdots,m$, we used the fact \eqref{d_cpfg_m} or \eqref{eq:TFATMD-interpration},
%\begin{align*}
%{d}_{\boldsymbol{c}}^{(p)}( \mathcal{F},\mathcal{G} )
&\le \Bigg( \sum_{i=1}^n{r \left( c_{\text{T}} T^{\theta}_{\text{TFA,}i} \right) ^p} +\sum_{i=1}^m{r \left( c_{\text{T}} T^{\theta}_{\text{TFA,}i} \right) ^p} +\sum_{j=1}^l\sum_{i=1}^n\Big[\mathbf{1}_{\theta}\left( h_j ,f_i  \right)   \left({d}_{\boldsymbol{c}}^{(p)}\left(  h_j ,f_i\right) \right) +\mathbf{1}_{\theta}\left( h_j ,g_i  \right)   \left({d}_{\boldsymbol{c}}^{(p)}\left(  h_j ,g_i\right) \right)\Big]^p\Bigg)^{\frac{1}{p}}\nonumber \\
  &\le \bigg( \sum_{i=1}^m{ r \left( c_{\text{T}} T^{\theta}_{\text{TFA,}i} \right) ^p}  +\sum_{j=1}^l\sum_{i=1}^m \mathbf{1}_{\theta}\left( h_j ,f_i  \right)  \Big( {d}_{\boldsymbol{c}}^{(p)}\left( h_j ,f_i\right) \Big) ^p \bigg)^{\frac{1}{p}}   \nonumber \\
 &~~~~~~+\bigg(\sum_{i=1}^n{ r \left( c_{\text{T}} T^{\theta}_{\text{TFA,}i} \right) ^p} +\sum_{j=1}^l\sum_{i=1}^n{\mathbf{1}_{\theta}\left( h_j ,g_i  \right)   \Big( {d}_{\boldsymbol{c}}^{(p)}\left( h_j ,g_i \right) \Big) ^p} \bigg)^{\frac{1}{p}} \label{eq:Star_ID_case3}\\
&= {d}_{\boldsymbol{c}}^{(p)}\left( \mathcal{H},\mathcal{F} \right)+{d}_{\boldsymbol{c}}^{(p)}\left( \mathcal{H},\mathcal{G} \right)   \nonumber
\end{align}
where \eqref{d_cpfg_m} was used in \eqref{eq:Case-3-prof}. % while Proposition 3 in \eqref{eq:Star_ID_case3}.
}
\end{table*}

%\end{appendices}

%\bibliographystyle{unsrt}
%\bibliographystyle{IEEEtran}
%\bibliography{TFoT_metric}


\end{document}